\documentclass[11pt,draftcls,onecolumn]{IEEEtran}
\ifCLASSINFOpdf
\else
\fi
\usepackage{tikz,amsmath,mathrsfs,amssymb,dsfont,amsthm}
\usepackage[normalem]{ulem}
\newcommand{\bbE}{\mathbb E}
\newcommand{\bbR}{\mathbb R}

\newcommand{\bfX}{{\mathbf X}}\newcommand{\bfx}{{\mathbf x}}
\newcommand{\bfY}{{\mathbf Y}}\newcommand{\bfy}{{\mathbf y}}
\newcommand{\bfW}{{\mathbf W}}\newcommand{\bfw}{{\mathbf w}}
\newcommand{\bfp}{{\mathbf p}}
\newcommand{\bfal}{\mbox{\boldmath $\alpha$}}

\newcommand{\calA}{{\mathcal A}}
\newcommand{\calE}{{\mathcal E}}
\newcommand{\calK}{{\mathcal K}}
\newcommand{\calM}{{\mathcal M}}

\newcommand{\calV}{{\mathcal V}}
\newcommand{\calX}{{\mathcal X}}
\newcommand{\calY}{{\mathcal Y}}
\newcommand{\calKhat}{\hat{\mathcal K}}

\newcommand{\scrP}{{\mathscr P}}
\newcommand{\scrW}{{\mathscr W}}

\newcommand{\dsone}{{\mathds 1}}

\newcommand{\sfK}{{\mathsf K}}

\newcommand{\Bernoulli}{\mathrm{Bernoulli}}
\newcommand{\Binomial}{\mathrm{Binomial}}
\newcommand{\Beta}{\mathrm{Beta}}
\newcommand{\col}{\mathrm{c}}
\newcommand{\emb}{\mathrm{e}}
\newcommand{\supp}{\mathrm{supp}}
\newcommand{\fair}{\mathrm{fair}}
\newcommand{\sym}{\mathrm{sym}}
\newcommand{\joint}{\mathrm{joint}}
\newcommand{\simple}{\mathrm{simple}}
\newcommand{\marking}{\mathrm{mark}}
\newcommand{\FP}{\mathrm{FP}}
\newcommand{\one}{\mathrm{one}}
\newcommand{\all}{\mathrm{all}}
\newcommand{\opt}{\mathrm{opt}}

\newtheorem{theorem}{Theorem}
\newtheorem{definition}{Definition}
\newtheorem{lemma}{Lemma}
\newtheorem{corollary}{Corollary}
\newtheorem{condition}{Condition}
\newtheorem{remark}{Remark}

\definecolor{gray}{rgb}{0.5,0.5,0.5}

\newcommand{\added}[1]{#1}
\newcommand{\changed}[1]{#1}
\newcommand{\removed}[1]{}
\newcommand{\removednosout}[1]{}
\newcommand{\archived}[1]{}

\begin{document}
%
\title{On the Saddle-point Solution and the Large-Coalition \changed{Asymptotics}\removed{Behavior} of Fingerprinting Games}
%
%
%

\author{Yen-Wei~Huang,~\IEEEmembership{Student~Member,~IEEE,}
        and~Pierre~Moulin,~\IEEEmembership{Fellow,~IEEE}
\thanks{The authors are with the Department of Electrical and Computer Engineering,
University of Illinois at Urbana-Champaign, Urbana, IL 61801 USA (e-mail:
huang37@illinois.edu; moulin@ifp.uiuc.edu). This work was supported by the National Science Foundation (NSF) under grants CCF 06-35137 and CCF 07-29061. This work was presented \changed{at ISIT 2009 and WIFS 2010}\removed{in part at the IEEE International Symposium on Information Theory (ISIT 2009), Seoul, Korea, Jun. 2009, and will be presented in part at the 2010 IEEE international Workshop on Information Forensics and Security (WIFS 2010), Seattle, WA, USA, Dec. 2010}.}}

\maketitle

\begin{abstract}
We study a fingerprinting game in which the number of colluders and the collusion channel are unknown. The encoder embeds fingerprints into a host sequence and provides the decoder with the capability to trace back pirated copies to the colluders.

Fingerprinting capacity has recently been derived as the limit value of a sequence of maximin games with mutual information as their payoff functions. However, these games generally do not admit saddle-point solutions and are very hard to solve numerically. Here under the so-called Boneh-Shaw marking assumption, we reformulate the capacity as the value of a single two-person zero-sum game, and show that it is achieved by a saddle-point solution.

If the maximal coalition size is $k$ and the fingerprinting alphabet is binary, we show that capacity decays quadratically with $k$. Furthermore, we prove rigorously that the asymptotic capacity is $1/(k^2 2 \ln2)$ and we confirm our earlier conjecture that Tardos' choice of the arcsine distribution asymptotically maximizes the mutual information payoff function while the interleaving attack minimizes it. Along with the asymptotic\added{s}\removed{ behavior}, numerical solutions to the game for small $k$ are also presented.
\end{abstract}


%
\IEEEpeerreviewmaketitle

%
%
%
%

\section{Introduction}

\IEEEPARstart{I}{n} view of the ubiquity of digital media and the development of sophisticated piracy tools, it has become essential to develop a reliable protection scheme for copyrighted content. Digital fingerprinting, in which the content distributor embeds a uniquely identified \emph{fingerprint} into each distributed copy, is an effective way to \changed{deter}\removed{prevent} unauthorized redistribution of the content.

Hundreds of years ago, people used the idea of fingerprinting in logarithm tables. Errors were added intentionally to insignificant \changed{decimals that are randomly selected}\removed{bits of random values}, with each copy having a unique set of modifications. If someone ever sold his copy illegally, the legal authority could easily trace the guilty owner (\emph{pirate}) by looking into the small errors.

Digital content (e.g., images, videos, audios, programs, etc.) can be protected using the same idea. One common approach is to embed fingerprints using digital watermarking techniques. Similar to the logarithm tables discussed above, the fingerprints should not impair the quality or the functionality of the contents. Most watermarking systems are even robust against attacks such as compression, digital-to-analog conversions, or intentional noise adding.

The most dangerous attack against fingerprinting, however, is a collusion attack. A group of experienced pirates can form a \emph{coalition}, detect the fingerprints by inspecting the \emph{marks} in each copy, and create a \emph{forgery} that has only weak traces of their fingerprints. For the logarithm table example, if the errors are sparse and chosen randomly, the coalition can easily correct these errors by comparing several different copies of the logarithm table. Notice that since the pirates cannot remove the errors in which all their copies coincide, it is still possible for the distributor to design the marks so that at least one of the pirates can be caught (with possibly certain risk of falsely accusing someone). Yet it should be apparent that if the size of coalition is large, it is very hard to trace back to the fingerprinted copies from which the forgery was generated.

To specify the type of manipulations the coalition is capable of, different models have been adopted in designing the collusion-resistant fingerprinting codes. The \emph{distortion constraint} is a natural model for fingerprinting on multimedia contents \cite{Moulin2003, Somekh-Baruch2005, Somekh-Baruch2007, Moulin2008b}. In this work, however, we adopt another setup introduced by Boneh and Shaw in \cite{Boneh1998}, called the \emph{marking assumption}, which is commonly used both in multimdeia fingerprinting \cite{Furon2009a} and software fingerprinting \cite{Anthapadmanabhan2008}. Under this setup, the fingerprint sequence that each user receives is represented by a string of marks. By comparing their available copies, the \changed{colluders}\removed{coalition} can modify the detected marks, but cannot modify those marks at which their copies agree. It should be noted that there exist several versions of the marking assumption specifying different strength of attacks the colluders can perform \cite{vSkori'c2008}, and our analysis is general and applies to all these variants.

\subsection{Previous Work}

One of the first designs of fingerprinting codes that are resistant to collusion attacks is presented by Boneh and Shaw \cite{Boneh1998}. It was shown in \cite{Boneh1998} that a deterministic binary fingerprinting code with zero probability of decoding error does not exist. Hence, it becomes necessary for the construction of the fingerprinting codes to use some form of randomization, where the random key is shared only between the encoder and the decoder. They also provided the first example of codes with vanishing error probability.

Tardos in 2003 \cite{Tardos2003} constructed fingerprinting codes of length at most $100 k^2 \ln(m/\epsilon)$ for $m$ users with error probability at most $\epsilon$ against $k$ pirates. This construction yields $k$-secure fingerprinting schemes \added{with $\epsilon$-error} of rate $\left[k^2 100 \ln (2/\epsilon)\right]^{-1}$. The same paper gave an $\Omega\left(k^2 \log(1/\epsilon)\right)$ bound on the length of any fingerprinting code with the above parameters. The constant $100$ in the length $100k^2 \ln(m/\epsilon)$ was subsequently improved by several papers \cite{Nuida2007,Blayer2008,vSkori'c2008,vSkori'c2008a}. Amiri and Tardos recently \cite{Amiri2009} further improved the rate by constructing a code based on a two-person zero-sum game.

A few researchers have also studied the problem from the information-theoretic point of view \cite{Moulin2003, Moulin2008b, Anthapadmanabhan2008, Amiri2009}. Here the main objective is to find the maximum achievable rate, or \emph{capacity}, of the fingerprinting system. We denote capacity by $C_k$ where $k$ is the maximum coalition size. For the binary \changed{alphabet}\removed{alpahbet}, Tardos' construction suggested that $C_k \geq (k^2 100 \ln 2)^{-1}$. Anthapadmanabhan et al. \cite{Anthapadmanabhan2008} proved that $C_k = O(1/k)$. Recently, Moulin \cite{Moulin2008b} provided the exact formula of capacity in a general setup that unifies the signal-distortion and Boneh-Shaw formulations of fingerprinting. The formula can be regarded as the limit value of a sequence of maximin games, which, however, is extremely difficult to evaluate in general.\looseness=-1

Two families of fingerprinting decoding scheme are also introduced in \cite{Moulin2008b}: simple decoding and joint decoding. The breakthrough of Tardos' randomized fingerprinting code in \cite{Tardos2003} and its subsequent works belong to the class of simple decoders. It falls short of the capacity-achieving goal: reliable performance is impossible at code rates greater than some value $C^\simple_k$ that is strictly less than the capacity $C_k$. Yet the simple and efficient algorithm makes it desirable for practical use. On the other hand, Amiri and Tardos' recent work \cite{Amiri2009} belonged to the joint decoding scheme. Although capacity-achieving, it is much more complex than the simple decoding scheme and is only useful when computation is not an issue. Another example of joint decoding is Dumer's work in \cite{Dumer2009}, where additional constraints are imposed in the analysis.

\subsection{Main Results}

Our work follows the Boneh-Shaw marking assumption. For both joint and simple decoding, we reformulate the maximum achievable rates of both schemes as the respective maximin values of two two-person zero-sum games. We further show that the maximin and minimax values of the games are always equal, and the values are achieved by saddle-point solutions.

In the binary alphabet case, new capacity bounds are provided in closed-form expressions. The ratio between the upper and lower bounds of the joint decoding scheme is $\pi^2/2$, while that of the simple decoding scheme is only $\pi^2/4$ (for large $k$). These bounds not only show that the binary fingerprinting capacity is in $\Theta(1/{k^2})$, but they also provide secure strategies for both players of the game. Numerical solutions for small $k$ are also presented in comparison with the bounds.

Asymptotic analysis for large coalitions is based on a mild regularity assumption. When $k$ is large, the fingerprinting game \changed{for}\removed{of} joint decoding approximates a continuous-kernel game, whose optimal-achieving strategies can be solved explicitly as the arcsine distribution and the interleaving attack. Finally, we give a higher level interpretation from the standpoint of statistical decision theory.

The outline of the paper is as follows: In Sec. \ref{sec:def}, we introduce our fingerprinting model and formally define fingerprinting capacity. The capacity formulas derived in \cite{Moulin2008b} are briefly reviewed and reformulated in Sec. \ref{sec:game}. Sec. \ref{sec:bin} and Sec. \ref{sec:bin_asym} are devoted to the binary alphabet case and Sec. \ref{sec:sum} gives a brief summary.

\subsection{Notation}\label{ssec:intro_notation}

We use capital letters to represent random variables, and lowercase letters to their realizations. Boldface denotes vectors, and calligraphic letters denote finite sets. For example, $\bfX \in \calX^n$ denotes a random vector $(X_1, \ldots, X_n)$, with each $X_i$ taking values in $\calX$.
The probability distribution of $\bfX$ is characterized by its distribution function $P_\bfX(\bfx) \triangleq Pr(X_1 \leq x_1, \ldots, X_n \leq x_n)$. If the distribution is discrete, we also describe it by its probability mass function (pmf) $p_\bfX(\bfx) \triangleq Pr(\bfX = \bfx)$. Otherwise if $P_\bfX$ has the form
$$P_\bfX(\bfx) = \int_{-\infty}^{x_1} \cdots \int_{-\infty}^{x_n} f_{\bfX}(\bfx) dx_1 \ldots dx_n$$
then we characterize the distribution by its probability density function (pdf) $f_\bfX$. 
Mathematical expectation of a function $g(\bfX)$ with respect to $P_\bfX$ is defined by 
$$\bbE_{P_\bfX}\left[(g(\bfX)\right] \triangleq \int g(\bfx) P(d\bfx).$$

The mutual information of $X$ and $Y$ is denoted by $I(X;Y) = H(X)-H(X|Y)$. Should the dependency on the underlying pmf's be explicit, we write the pmf's as subscripts, e.g. $H_{p_X}(X)$ and $I_{p_X p_{Y|X}}(X;Y)$.
Given a pair of sequences $(\bfx,\bfy)$, we denote by $I(\bfx;\bfy)$ the empirical mutual information of the joint pmf $p_{\bfx\bfy}$.
We also denote the binary entropy function by $h(p) \triangleq -p\log p - (1-p)\log (1-p)$ and $h(\bfp) = \left(h(p_1), \ldots, h(p_n)\right)'$.
The Kullback-Leibler divergence between two pmf's $p$ and $q$ is denoted by $D(p \parallel q)$, and the Kullback-Leibler divergence between two Bernoulli random variables with respective expectations $p$ and $q$ is denoted by $d(p\|q) \triangleq p\log\frac{p}{q}+(1-p)\log\frac{1-p}{1-q}$, where $\log$ denotes base 2 logarithm and $\ln$ denotes natural logarithm throughout the paper.

\added{Sequences are denoted by $\langle\cdot\rangle$.} The size or cardinality of a finite set $\calA$ is denoted by $|\calA|$.
The indicator function of a subset $\calA$ of a set $\calX$ is a function $\dsone_\calA: \calX \rightarrow \{0,1\}$ defined as
$$\dsone_\calA(x) = \left\{ \begin{array}{ll} 1 & \textrm{if } x \in \calA\\ 0 & \textrm{if } x \notin \calA \end{array} \right..$$
The power set of a finite set $\calX$, denoted by $2^\calX$, is the set of all subsets of $\calX$, including the empty set and $\calX$ itself.
The support of a probability distribution $P$, denoted by $\supp(P)$, is the smallest set whose complement has probability zero. The support of a family $\scrP$ of probability distributions, denoted by $\supp(\scrP)$, is the union of the support of each probability distribution in the family, i.e., $\bigcup_{P \in \scrP} \supp(P)$.

Asymptotic notations are defined as follows: Suppose $f(k)$ and $g(k)$ are two functions defined on positive real numbers. We say $f(k) = O(g(k))$ if $\exists c_1>0, k_1 > 0$ such that $f(k) \leq c_1 g(k), \forall k \geq k_1$. Also, $f(k) = \Omega(g(k))$ if $\exists c_2>0, k_2 > 0$ such that $f(k) \geq c_2 g(k), \forall k \geq k_2$. We write $f(k) = \Theta(g(k))$ if $f(k) = O(g(k))$ and $f(k) = \Omega(g(k))$. The expression $f(k) = o(g(k))$ or $f(k) = \omega(g(k))$ means that $f(k)/g(k)$ tends to 0 or $\infty$ respectively.
The shorthand $f \sim g$, $f \gtrsim g$, and $f \lesssim g$ denote the asymptotic relations $\lim_{k \rightarrow \infty} \dfrac{f(k)}{g(k)} = 1$, $\liminf_{k \rightarrow \infty} \dfrac{f(k)}{g(k)} \geq 1$, and $\limsup_{k \rightarrow \infty} \dfrac{f(k)}{g(k)} \leq 1$ respectively.

\section{Fingerprinting codes and capacity}\label{sec:def}

\subsection{Overview}\label{ssec:def_overview}

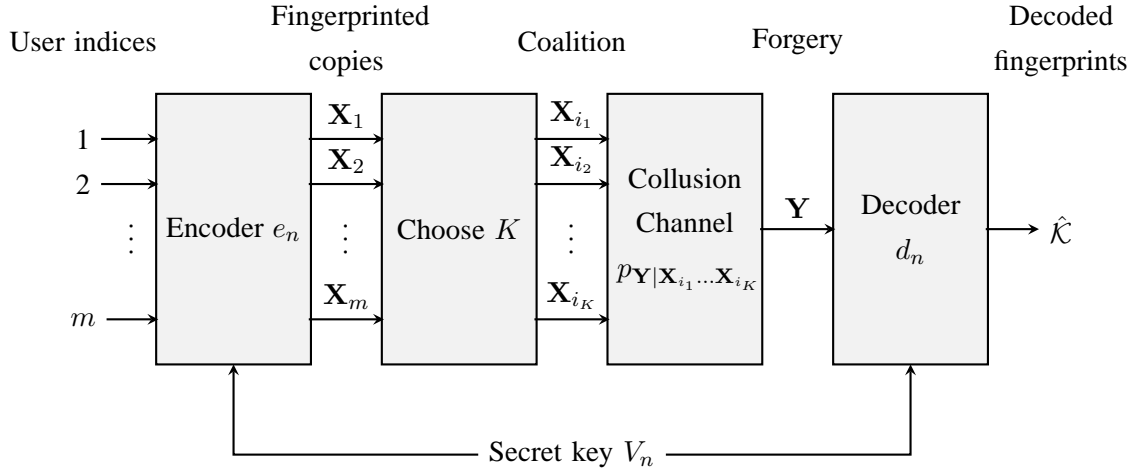
\begin{figure}[t]
\centering
\tikzstyle{block} = [draw, rectangle, text width=1.8cm, text centered, fill=black!5, minimum height=3.6cm, minimum width=2cm, scale=1]

\begin{tikzpicture}[>=stealth, thick, scale=1]

\node at (2,0) [block] (encoder)    {Encoder $e_n$};
\node at (5,0) [block] (multiplex)  {Choose $K$};
\node at (8,0) [block] (collusion)  {Collusion Channel $p_{\bfY|\bfX_{i_1}\ldots \bfX_{i_K}}$};
\node at (11,0)[block] (decoder)    {Decoder $d_n$};

\node at (6.5,-3) (key) {Secret key $V_n$};
\node at (0,2.5) [text width=2cm, text centered]() {User indices};
\node at (3.5,2.5) [text width=2cm, text centered]() {Fingerprinted copies};
\node at (6.5,2.5) [text width=2cm, text centered]() {Coalition};
\node at (9.5,2.5) [text width=2cm, text centered]() {Forgery};
\node at (13,2.5) [text width=2cm, text centered]() {Decoded fingerprints};

\node at (0,1.2)  (u1) {1};
\node at (0,.6) (u2) {2};
\node at (0,-1.2) (um) {$m$};
\node at (13,0) (calKhat) {$\calKhat$};

\draw[->] (u1) -- (1,1.2);
\draw[->] (u2) -- (1,.6);
\draw[->] (um) -- (1,-1.2);

\draw[->] (3,1.2)  -- node[above]{$\bfX_1$} (4,1.2);
\draw[->] (3,.6)   -- node[above]{$\bfX_2$} (4,.6);
\draw[->] (3,-1.2) -- node[above]{$\bfX_m$} (4,-1.2);

\draw[->] (6,1.2)  -- node[above]{$\bfX_{i_1}$} (7,1.2);
\draw[->] (6,.6)   -- node[above]{$\bfX_{i_2}$} (7,.6);
\draw[->] (6,-1.2) -- node[above]{$\bfX_{i_K}$} (7,-1.2);

\draw[->] (9,0)  -- node[above]{$\bfY$} (10,0);
\draw[->] (12,0) -- (calKhat);

\draw[->] (key) -| (2,-1.8);
\draw[->] (key) -| (11,-1.8);

\node at (.65,-.1) [rotate=90] {$\cdots$};
\node at (3.5,-.1) [rotate=90] {$\cdots$};
\node at (6.5,-.1) [rotate=90] {$\cdots$};

\end{tikzpicture}
\bigskip
\caption{Fingerprinting system model under the marking assumption}
\label{fig:model}
\end{figure}

The model for our fingerprinting system is shown in Fig. \ref{fig:model}. Let $\calX = \{0,1,\ldots,q-1\}$ denote a size-$q$ fingerprint alphabet and let $\calM = \{1, \ldots, m\}$ denote the set of user indices. An $(n,m)$ fingerprinting code $(e_n, d_n)$ over $\calX$ consists of an encoder and a decoder. The encoder
\begin{equation}
e_n: \calM \times \calV_n \rightarrow \calX^n
\end{equation}
assigns user $i$ a length-$n$ fingerprint $\bfX_i$, where $\calV_n$ is the alphabet of the secret key $V_n$, which is a random variable whose realization is known to the encoder and the decoder, but unknown to the pirates.

\added{We denote by} $\calK \triangleq \{i_1, \ldots, i_K\}$\removed{is} the index set of a coalition of $K$ pirates\removed{,} and $\bfX_\calK = \{\bfX_i: i \in \calK\}$ \changed{as}\removed{are} the fingerprints available to them. The collusion channel produces the forgery $\bfY \in \calY^n$ according to some conditional probability mass function (pmf) $p_{\bfY|\bfX_\calK}$. The Boneh-Shaw marking assumption is imposed on $p_{\bfY|\bfX_\calK}$, which allows the colluders to change only the symbols at the positions where they find differences.

Not knowing the actual collusion channel $p_{\bfY|\bfX_\calK}$, the decoder
\begin{equation}
d_n: \calY \times \calV_n \rightarrow 2^\calM
\end{equation}
produces an estimate $\calKhat$ of the coalition. Note that the actual number of pirates $K$ is known neither to the encoder nor to the decoder, so the code design is based on a nominal coalition size $k$. Also note that the empty set $\emptyset$ is an admissible decoder output, which enables us not to accuse any user when no enough evidence is available to the decoder, especially when the actual $K$ is larger than $k$.

\subsection{Randomized Fingerprinting Codes}\label{ssec:def_random_code}

The formal definition of an ensemble of fingerprinting codes is as follows.

\begin{definition}\label{def:ensemble}
A fingerprinting ensemble $(E_n, D_n)$ is formed by the fingerprinting embedder randomly choosing from a family $\{ e_n(\cdot, v_n), d_n(\cdot, v_n), v_n \in \calV_n \}$ of $(n,m)$ fingerprinting codes according to some probability distribution on the set $V_n$ of keys.
\end{definition}

We assume that the family of fingerprinting codes and the probability distribution on $V_n$ are known to the public, but the realization $v_n$ is only known to the encoder and the decoder.

As shown in \cite{Moulin2008b}, it suffices to consider the following two-phase fingerprinting construction and joint/simple decoding scheme in studying capacity. The secret key $V_n$ shared by the encoder and the decoder in this scheme is the set of random variables $\{\bfW_j\}_{j=1}^n \cup (X_{i,j})_{m \times n}$.

\subsubsection{Encoding Scheme}

Let $P_\bfW$ be a probability distribution on the $(q-1)$-dimensional simplex
\begin{equation}
\scrW^q \triangleq \left\{ \bfw \in \bbR^q: \sum_{x=0}^{q-1} w_x = 1 \textrm{ and } 0 \leq w_x \leq 1, x \in \calX \right\}.
\end{equation}
A sequence of auxiliary ``time-sharing'' random variables $\{ \bfW_j \}_{j=1}^n$ is drawn independent and identically from the distribution $P_\bfW$. For each $j \in \{1, \ldots, n\}$, $\{X_{i,j}\}_{i=1}^m$ are $m$ independent and identically distributed random variables constructed from a categorical distribution\footnote{The categorical distribution is a special case of the multinomial distribution with the number of trials set to 1.} with parameter $\bfW_j$, i.e.,
\begin{equation}
Pr \left( X_{1,j} = x_1, \ldots, X_{m,j} = x_m | \bfW_j = \bfw \right) = \prod_{i=1}^m w_{x_i}, \quad \forall x_1, \ldots, x_m \in \calX.
\end{equation}

In general, there is no constraint on the choice of the embedding distribution $P_\bfW$, which means that we can choose it from the class of all probability distributions on $\scrW^q$, denoted by $\scrP_\bfW$. However, we may want to limit $P_\bfW$ to a subclass $\scrP^\emb$ of $\scrP_\bfW$ in some applications. For instance, Nuida et al. \cite{Nuida2007} limited $P_\bfW$ to be discrete with a finite spectrum. Furon and Perez-Freire \cite{Furon2009a} studied the case when $P_W$ is the arcsine distribution (defined below) or the uniform distribution over the unit interval for binary fingerprinting codes, in which $\scrP^\emb$ is just a singleton.

In most of our results we require $\scrP^\emb$ to be compact\removed{, which is true in general}. In some results we also require the following condition. \added{Note that $\scrP^\emb$ satisfying (\ref{eqn:prob}) is compact.}

\begin{condition}\label{cond:prob}
$\scrP^\emb$ \changed{coincides with}\removed{is} the class of all probability distributions on $\supp(\scrP^\emb)$, i.e.,
\begin{equation}\label{eqn:prob}
\scrP^\emb = \left\{ P_\bfW \in \scrP_\bfW: \supp(P_\bfW) \subseteq \supp(\scrP^\emb) \right\}.
\end{equation}
\end{condition}

Analogous to the symbol-symmetric fingerprinting codes proposed by \v{S}kori\'{c} et al. \cite{vSkori'c2008}, it is intuitively reasonable to adopt a probability distribution $P_\bfW$ that is invariant to permutations of the symbols. Formally, let $\pi$ be a permutation of $\calX$ and define

\begin{equation}
P_\bfW^\pi(w_0, \ldots, w_{q-1}) \triangleq P_\bfW(w_{\pi(0)}, \ldots, w_{\pi(q-1)}).
\end{equation}
Then we have
\begin{definition}
An embedding distribution $P_\bfW$ is {\bf symbol-symmetric} if
\begin{equation}
P_\bfW^\pi = P_\bfW, \quad \forall \pi.
\end{equation}
\end{definition}

\begin{definition}
A subset $\scrP^\emb$ of $\scrP_\bfW$ is said to be \emph{symbol-symmetric} if
\begin{equation}
P_\bfW \in \scrP^\emb \Rightarrow\ P_\bfW^\pi \in \scrP^\emb, \quad \forall \pi.
\end{equation}
\end{definition}

\begin{definition}
The symbol-symmetric subclass of an embedding class $\scrP^\emb$ is defined by
\begin{equation}
\scrP^\emb_\sym = \left\{ P_\bfW \in \scrP^\emb: P_\bfW \textrm{ is symbol-symmetric} \right\}.
\end{equation}
\end{definition}

The optimality of limiting $P_\bfW$ to $\scrP^\emb_\sym$ will be discussed in the next section.

\subsubsection{Decoding Scheme}

We briefly review Moulin's two decoding schemes proposed in \cite{Moulin2008b}: the simple decoder tests candidate fingerprints one by one, while the joint decoder utilizes a joint decoding rule. The simple decoder evaluates the empirical mutual information $I(\bfx_i;\bfy|\bfw)$ for each user $i$. A threshold $\eta^\simple$ is chosen and user $i$ is accused if and only if $I(\bfx_i;\bfy|\bfw) > \eta^\simple$. If $I(\bfx_i;\bfy|\bfw) \leq \eta^\simple$ for all $i \in \calM$, then $\hat{\calK} = \emptyset$. The joint decoder evaluates the following score for each coalition $\calA \subseteq \calM$
\begin{equation}
S(\calA) = \left\{\begin{array}{ll}
0, &\textrm{ if } \calA = \emptyset\\
I(\bfx_\calA;\bfy|\bfw) - |\calA|\eta^\joint, &\textrm{ otherwise }
\end{array}\right.
\end{equation}
where $\eta^\joint$ is a threshold. The set $\calA$ that has the largest score is then accused. With the parameters $\eta^\simple$ and $\eta^\joint$, both decoders allow to tune the trade-offs between false positive and false negative error probabilities.

It is shown in \cite{Moulin2008b} that the joint decoding scheme achieves capacity while the simple decoding scheme has a smaller maximum achievable rate. However, \changed{the computational complexity of a joint decoder is generally vastly greater than that of a simple decoder.}\removed{the joint decoder is generally computationally infeasible in practice. The simple decoder, albeit suboptimal, is more practical.}

\subsection{Collusion Channel}\label{ssec:def_colchan}

Upon receiving the $K$ fingerprinting copies $\{\bfX_{i_1}, \ldots, \bfX_{i_K}\}$, the pirates attempt to generate a forgery $\bfY$ subject to the marking assumption. A coordinate $j$ is called \emph{undetectable} if
\begin{equation}
x_{i_1,j} = x_{i_2,j} = \cdots = x_{i_K,j}
\end{equation}
and is called detectable otherwise. The Boneh-Shaw marking assumption states that for any forgery $\bfy$ generated by the coalition, we have $y_j = x_{i_1,j}$ for every undetectable coordinate $j$.

We assume the collusion channel $p_{\bfY|\bfX_\calK}$ adopted by the pirates is memoryless, that is, 
\begin{equation}
p_{\bfY|\bfX_\calK}(\bfy|\bfx_\calK) = \prod_{j=1}^n p_{Y|X_\calK}(y_j|x_{\calK,j}).
\end{equation}
As exploited in \cite{Moulin2008b}, the memoryless restriction can be relaxed without changing the fingerprinting capacity. For simplicity we impose this constraint so the colluders' strategy is limited to the choice of the single-lettered channel $p_{Y|X_\calK}$. Denoted by $\scrP_\marking$, the class of attacks satisfying the marking assumption can be written as
\begin{equation}\label{eqn:marking}
\scrP_\marking = \{ p_{Y|X_\calK}: p_{Y|X_\calK}(y|x_\calK) = 1 \textrm{ if } y = x_{i_1} = x_{i_2} = \cdots = x_{i_K} \}.
\end{equation}

Several variants of the marking assumption appear in the literature. Suppose $\scrP^\col$ denotes the class of admissible channels $p_{Y|X_\calK}$. Some restrict the coalition to use only the symbols available to them, where $\scrP^\col = \{p_{Y|X_\calK} \in \scrP_\marking: p_{Y|X_\calK}(y|x_\calK) = 0 \textrm{ if } y \neq x_i, \forall i \in \calK \}$ (usually called the \emph{restricted digit model}). Others relax the alphabet $\calY$ to $\calX \cup \{\textrm{`?'}\}$, which allows the pirates to put \emph{erasures} at detectable coordinates (usually called the \emph{general digit model}). Here we consider a general $\scrP^\col$ with $\scrP^\col \subseteq \scrP_\marking$ being compact.

Given a single-lettered collusion channel $p_{Y|X_\calK}$, consider the user-permuted collusion channel
\begin{equation}
p_{Y|X_{\pi(\calK)}}(y|x_{i_1}, \ldots, x_{i_K}) \triangleq p_{Y|X_\calK}(y|x_{\pi(i_1)}, \ldots, x_{\pi(i_K)})
\end{equation}
where $\pi$ is a permutation of $\calK$. We say that $p_{Y|X_\calK}$ is \emph{user-symmetric} if
\begin{equation}
p_{Y|X_{\pi(\calK)}} = p_{Y|X_\calK}, \quad \forall \pi.
\end{equation}
A subset $\scrP^\col$ of $\scrP_\marking$ is said to be user-symmetric if
\begin{equation}
p_{Y|X_\calK} \in \scrP^\col \Rightarrow\ p_{Y|X_{\pi(\calK)}} \in \scrP^\col, \quad \forall \pi.
\end{equation}

Note that in general not all elements of such $\scrP^\col$ are user-symmetric. The user-symmetric subclass of a collusion class $\scrP^\col$ that consists of user-symmetric collusion channels is defined by
\begin{equation}
\scrP^\col_\fair = \{ p_{Y|X_\calK} \in \scrP^\col: p_{Y|X_\calK} \textrm{ is user-symmetric} \}.
\end{equation}

Symbol-symmetry can also be defined for collusion channels. Let $\pi$ be a permutation of $\calX$ and define
\begin{equation}
p_{Y|X_\calK}^\pi(y|x_{i_1}, \ldots, x_{i_K}) \triangleq p_{Y|X_\calK}(y|\pi(x_{i_1}), \ldots, \pi(x_{i_K})).
\end{equation}
Then we have
\begin{definition}
A collusion channel $p_{Y|X_\calK}$ is {\bf symbol-symmetric} if
\begin{equation}
p_{Y|X_\calK}^\pi = p_{Y|X_\calK}, \quad \forall \pi.
\end{equation}
\end{definition}

\begin{definition}
A subset $\scrP^\col$ of $\scrP_\marking$ is said to be \emph{symbol-symmetric} if
\begin{equation}
p_{Y|X_\calK} \in \scrP^\col \Rightarrow\ p_{Y|X_\calK}^\pi \in \scrP^\col, \quad \forall \pi.
\end{equation}
\end{definition}

\begin{definition}
The symbol-symmetric subclass of a collusion class $\scrP^\col$ is defined by
\begin{equation}
\scrP^\col_\sym = \left\{ p_{Y|X_\calK} \in \scrP^\col: p_{Y|X_\calK} \textrm{ is symbol-symmetric} \right\}.
\end{equation}
\end{definition}

\subsection{Error Probabilities and Capacity}\label{ssec:err_cap}

Under fingerprinting ensemble $(E_n, D_n)$, nominal coalition size $k$ (not necessarily equal to the true coalition size $K$), and collusion channel $p_{Y|X_\calK}$, we consider the following error probabilities:
\begin{itemize}
\item The probability of false positives (accusing an innocent user):
\begin{equation}\label{eqn:FPerror}
P^{\FP}_e(E_n, D_n, p_{Y|X_\calK}) = Pr \left(\calKhat \setminus \calK \neq \emptyset \right).
\end{equation}
\item The probability of failing to catch any single pirate:
\begin{equation}\label{eqn:oneerror}
P^{\one}_e(E_n, D_n, p_{Y|X_\calK}) = Pr \left(\calKhat \cap \calK = \emptyset \right).
\end{equation}
\item The probability of failing to catch the full coalition:
\begin{equation}\label{eqn:allerror}
P^{\all}_e(E_n, D_n, p_{Y|X_\calK}) = Pr \left(\calK \nsubseteq\calKhat \right).
\end{equation}
\end{itemize}

The error probabilities above can be written explicitly as
\begin{equation}
P_e(E_n, D_n, p_{Y|X_\calK}) = \sum_{\bfx_\calM, \bfy} \prod_{j=1}^n \int P_\bfW (d\bfw_j) \left( \prod_{i=1}^m w_{x_{i,j}} \right) p_{Y|X_\calK}(y_j|x_{\calK,j}) \dsone_\calE
\end{equation}
where the error event $\calE$ is given by $\calE^\FP = \left\{ d_n(\bfy,v_n) \setminus \calK \neq \emptyset \right\}$, $\calE^\one = \left\{ d_n(\bfy,v_n) \cap \calK = \emptyset \right\}$, and $\calE^\all = \left\{ \calK \nsubseteq d_n(\bfy,v_n) \right\}$, when $P_e$ is given by (\ref{eqn:FPerror}), (\ref{eqn:oneerror}), and (\ref{eqn:allerror}) respectively. The worst-case error probability for a collusion class is given by
\begin{equation}
P_{e,k}(E_n, D_n, \scrP^\col) = \max_{\substack{\calK \subseteq \calM\\ |\calK| \leq k}} ~ \max_{p_{Y|X_\calK} \in \scrP^\col} P_e\left(E_n, D_n, p_{Y|X_\calK}\right).
\end{equation}

Having defined the error probabilities of the randomized fingerprinting scheme, we now define the notion of capacity.

\begin{definition}
A rate $R$ is achievable for embedding class $\scrP^\emb$, collusion channel $\scrP^\col$, and size-$k$ coalitions under the {\bf detect-one} criterion if there exists a sequence of fingerprinting ensembles $(F_n, G_n)$ generated by $P_\bfW \in \scrP^\emb$ for $m = \lceil 2^{nR} \rceil$ users such that both $P^\FP_{e,k}(E_n, D_n, \scrP^\col)$ and $P^\one_{e,k}(E_n, D_n, \scrP^\col)$ vanish as $n$ tends to infinity.
\end{definition}

\begin{definition}
A rate $R$ is achievable for embedding class $\scrP^\emb$, collusion channel $\scrP^\col$, and size-$k$ coalitions under the {\bf detect-all} criterion if there exists a sequence of fingerprinting ensembles $(F_n, G_n)$ generated by $P_\bfW \in \scrP^\emb$ for $m = \lceil 2^{nR} \rceil$ users such that both $P^\FP_{e,k}(E_n, D_n, \scrP^\col)$ and $P^\all_{e,k}(E_n, D_n, \scrP^\col)$ vanish as $n$ goes to infinity.
\end{definition}

\begin{definition}
\begin{sloppypar}
Fingerprinting capacities $C^\one_k(\scrP^\emb,\scrP^\col)$ and $C^\all_k(\scrP^\emb,\scrP^\col)$ are the suprema of all achievable rates with respect to the detect-one and
detect-all criteria, respectively.
\end{sloppypar}
\end{definition}

\begin{remark}
When the embedding class $\scrP^\emb$ is a singleton $\{ P_\bfW \}$ or the collusion class $\scrP^\col$ is a singleton $\{p_{Y|X_\calK}\}$, we denote the corresponding capacities as $C_k(P_\bfW, \cdot)$ and $C_k(\cdot, p_{Y|X_\calK})$ respectively, which is a slight abuse of notation.
\end{remark}   
\section{Mutual information games}\label{sec:game}

In this section we first review the mutual information games associated with both the joint and the simple decoding schemes in \cite{Moulin2008b}. We show how these games can be simplified under the marking assumption, and we show the existence of saddle-point solutions.

\subsection{Mutual Information Game \changed{for}\removed{of} Joint Decoder}\label{ssec:game_joint}

We use the special symbol $\sfK$ to denote the set $\{ 1, 2, \ldots, k \}$, where $k$ is the nominal coalition size introduced in Sec. \ref{ssec:def_overview}. To present the capacity formula, we first introduce the following \changed{setup}\removed{experiment}: for a fixed embedding class $\scrP^\emb$, let 
\begin{equation}
\scrP^\emb_l \triangleq \left\{ p_\bfW \in \scrP^\emb: \left|\supp(p_\bfW)\right| \leq l \right\}
\end{equation}
be the class of probability distributions with finite spectrum composed of no more than $l$ points of the $(q-1)$-dimensional simplex $\scrW^q$. A random variable $\bfW$ is drawn from some $p_\bfW \in \scrP^\emb_l$, and $\{X_i\}_{i=1}^k$ are independent and identically distributed with categorical distribution with parameter $\bfW$, i.e.,
\begin{equation}
p_{X_\sfK|\bfW}(x_\sfK|\bfw) = \prod_{i=1}^k p_{X|\bfW}(x_i|\bfw)
\end{equation}
where
$$p_{X|\bfW}(x|\bfw) = w_x, \quad x \in \calX.$$
The collusion class $\scrP^\col$ is the set of all feasible channels $p_{Y|X_\sfK}$. Let
\begin{equation}\label{eqn:maximin_l}
C_k^{\joint,l}(\scrP^\emb, \scrP^\col) = \max_{p_\bfW \in \scrP^\emb_l} ~ \min_{p_{Y|X_\sfK} \in \scrP^\col} ~ \frac{1}{k} I(X_\sfK; Y|\bfW).
\end{equation}
The following theorem summarizes the main results of fingerprinting capacity proposed in \cite{Moulin2008b} under the marking assumption.

\begin{theorem}\label{thm:cap_one_all}
Assume that Condition \ref{cond:prob} is satisfied and $\scrP^\col$ is compact. Then
\begin{enumerate}
\item $C_k^\all(\scrP^\emb, \scrP^\col) \leq C^\one_k(\scrP^\emb, \scrP^\col)$. In particular, the detect-all fingerprinting capacity $C_k^\all(\scrP^\emb, \scrP_\marking)$ under the marking assumption is zero.
\item Suppose further that $\scrP^\col$ is user-symmetric, then 
\begin{equation}
C^\one_k(\scrP^\emb, \scrP^\col) = C^\one_k(\scrP^\emb, \scrP^\col_\fair) = C^\all_k(\scrP^\emb, \scrP^\col_\fair) = \lim_{l \rightarrow \infty} C^{\joint,l}_k(\scrP^\emb, \scrP^\col).
\end{equation}
\end{enumerate}
\end{theorem}

Theorem \ref{thm:cap_one_all} states that, the detect-all capacity can never exceed the detect-one capacity, which is no surprise since it \changed{can only be harder}\removed{is a much harder task} for the decoder to detect all the pirates than to detect only one of the pirates. However, if the collusion channel is user-symmetric, which we can intuitively think of as the case when each colluder ``contributes'' the same number of samples to the forgery (hence the term ``fair''), then the detect-one and the detect-all capacities are the same.

Now since the detect-all capacity is null under the marking assumption, we will in the rest of the paper refer to the detect-one capacity $C^\one_k(\scrP^\emb, \scrP^\col)$, denoted by $C^\joint_k(\scrP^\emb, \scrP^\col)$, as the joint fingerprinting capacity for embedding class $\scrP^\emb$ and collusion channel $\scrP^\col$.

In the game-theoretic point of view, $C_k^{\joint,l}$ is the maximin value of a two-person zero-sum game for each $l$. Observe that the sequence $\langle C_k^{\joint,l} \rangle_{l=1}^\infty$ is nondecreasing since $\langle \scrP^\emb_l \rangle_{l=1}^\infty$ is nondecreasing \added{(i.e. $\scrP_1^\emb \subseteq \scrP_2^\emb \subseteq \cdots$)}. Thus the game can be interpreted as the following: the maximizer, the fingerprint embedder, picks $p_\bfW$ with an increasing flexibility in the support size, while the minimizer, the coalition, counters the embedder's choice for each $l$ by minimizing the mutual information payoff function. Fingerprinting capacity is the limit value of the sequence of maximin games.

However, the maximin game of (\ref{eqn:maximin_l}) is in general very difficult to solve  even for small values of $l$ since a saddle-point solution cannot be guaranteed. For the binary alphabet ($q = 2$) and $l = 1$, we can derive the maximin value as
\begin{equation}\label{eqn:C_expo}
C^{\joint,1}_k(\scrP_\bfW,\scrP_\marking) = \frac{1}{k}2^{-(k-1)}
\end{equation}
which is not achieved by a saddle-point solution when $k>2$. Also, this is a very loose lower bound on $C^\joint_k(\scrP_\bfW,\scrP_\marking)$ for large $k$ comparing to the $\Theta(k^{-2})$ bound we will show in Sec. \ref{ssec:bin_bounds}.

\subsection{Mutual Information Game \changed{for}\removed{of} Simple Decoder}\label{ssec:game_simple}

As mentioned in Sec. \ref{ssec:def_random_code}, computationally \removed{the} joint decoding is too complex. Thus it is also interesting to study the maximum achievable rate of the simple decoding scheme.

\begin{theorem}
Assume that Condition \ref{cond:prob} is satisfied and $\scrP^\col$ is compact and user-symmetric. Let
\begin{equation}
C_k^{\simple,l}(\scrP^\emb, \scrP^\col) = \max_{p_\bfW \in \scrP^\emb_l} ~ \min_{p_{Y|X_\sfK} \in \scrP^\col_\fair} ~ I(X_1; Y|\bfW)
\end{equation}
for $l \geq 1$ and let
\begin{equation}
C^\simple_k(\scrP^\emb, \scrP^\col) = \lim_{l \rightarrow \infty} C_k^{\simple,l}(\scrP^\emb, \scrP^\col).
\end{equation}
Then all rates below $C^\simple_k(\scrP^\emb, \scrP^\col)$ are achievable by the simple decoding scheme for embedding class $\scrP^\emb$, collusion channel $\scrP^\col$, and size-$k$ coalitions under the detect-one criterion.
\end{theorem}
\begin{corollary}
For $\scrP^\emb$ satisfying Condition \ref{cond:prob} and compact and user-symmetric $\scrP^\col$, we have
\begin{equation}
C^\simple_k(\scrP^\emb, \scrP^\col) \leq C^\joint_k(\scrP^\emb, \scrP^\col).
\end{equation}
\end{corollary}

\begin{proof}
See \cite{Moulin2008b}.
\end{proof}

Although we do not have a notion of capacity for the quantity $C^\simple_k$, it will be referred to as the ``simple'' fingerprinting capacity as opposed to the joint fingerprinting discussed in the previous subsection.

\subsection{Two-Person Zero-Sum Games of Fingerprinting Capacity}\label{ssec:game_analysis}

To establish the desired saddle-point property, we first reformulate both the joint and the simple fingerprinting capacities as the respective values of the following two fingerprinting maximin games.

\begin{theorem}\label{thm:single_game}
Assume that Condition \ref{cond:prob} is satisfied and $\scrP^\col$ is compact and user-symmetric. Then
\begin{equation}\label{eqn:maximin_joint}
C_k^\joint(\scrP^\emb, \scrP^\col) = \max_{P_\bfW \in \scrP^\emb} ~ \min_{p_{Y|X_\sfK} \in \scrP^\col} ~ \frac{1}{k} I(X_\sfK; Y|\bfW)
\end{equation}
and
\begin{equation}\label{eqn:maximin_simple}
C_k^\simple(\scrP^\emb, \scrP^\col) = \max_{P_\bfW \in \scrP^\emb} ~ \min_{p_{Y|X_\sfK} \in \scrP^\col_\fair} ~ I(X_1; Y|\bfW).
\end{equation}
\end{theorem}

\begin{proof}
\begin{sloppypar}
Let
\begin{equation}
I^\joint_k(\bfw, p_{Y|X_\sfK}) = \frac{1}{k} I_{p_{Y|X_\sfK}}(X_\sfK; Y|\bfW = \bfw)
\end{equation}
and let 
\begin{equation}
I^\simple_k(\bfw, p_{Y|X_\sfK}) = I_{p_{Y|X_\sfK}}(X_1; Y|\bfW = \bfw).
\end{equation}
Then the payoff functions of (\ref{eqn:maximin_joint}) and (\ref{eqn:maximin_simple}) become $\bbE_{P_\bfW} \left[ I^\joint_k(\bfW, p_{Y|X_\sfK}) \right]$ and $\bbE_{P_\bfW} \left[ I^\simple_k(\bfW, p_{Y|X_\sfK}) \right]$ respectively. Denote also the right-hand sides of (\ref{eqn:maximin_joint}) and (\ref{eqn:maximin_simple}) by $\tilde{C}_k^\joint(\scrP^\emb, \scrP^\col)$ and $\tilde{C}_k^\simple(\scrP^\emb, \scrP^\col)$ respectively. The inequality
\end{sloppypar}
\begin{equation}\label{eqn:CleqCtilde}
C_k(\scrP^\emb, \scrP^\col) \leq \tilde{C}_k(\scrP^\emb, \scrP^\col)
\end{equation}
follows directly from the fact that $\scrP^\emb_l \subseteq \scrP^\emb$ for each $l$. Now let the optimal achieving distributions for (\ref{eqn:maximin_joint}) or (\ref{eqn:maximin_simple}) be $P_\bfW^{(k)}$ and $p_{Y|X_\sfK}^{(k)}$. Then by completeness of $\scrP^\emb$, there exists a sequence of distributions $\langle p_\bfW^l \rangle_{l=1}^\infty$ with $p_\bfW^l \in \scrP^\emb_l$ that converges in distribution to $P_\bfW^{(k)}$. Both the functions $I^\joint_k$ and $I^\simple_k$ are bounded and continuous with respect to $\bfw$. By \cite[p.249 Theorem 1]{Feller1968} we have
\begin{equation}
\lim_{l \rightarrow \infty} \bbE_{p_\bfW^l} \left[ I_k(\bfW, p^{(k)}_{Y|X_\sfK}) \right] = \bbE_{P^{(k)}_\bfW} \left[ I_k(\bfW, p^{(k)}_{Y|X_\sfK}) \right]
\end{equation}
and thus
\begin{equation}\label{eqn:CgeqCtilde}
C_k(\scrP^\emb, \scrP^\col) \geq \lim_{l \rightarrow \infty} \bbE_{p_\bfW^l} \left[ I_k(\bfW, p^{(k)}_{Y|X_\sfK}) \right] = \tilde{C}_k(\scrP^\emb, \scrP^\col).
\end{equation}
Combining (\ref{eqn:CleqCtilde}) and (\ref{eqn:CgeqCtilde}) yields (\ref{eqn:maximin_joint}) and (\ref{eqn:maximin_simple}).
\end{proof}

Theorem \ref{thm:single_game} shows that the joint and simple capacities are the maximin values of two single two-person zero-sum games. Note that the theorem only specifies the capacities when the embedding class $\scrP^\emb$ satisfies Condition \ref{cond:prob}. With slight modification of the proofs in \cite{Moulin2008b}, it can be shown that (\ref{eqn:maximin_joint}) and (\ref{eqn:maximin_simple}) still hold for any compact $\scrP^\emb$. Furthermore, we can show that the maximin and minimax values of the games are equal in general, and there are always saddle-point strategies for both players of the games.

We define the \changed{minimax}\removed{following} values associated \added{with} the above games:

\begin{definition}
The minimax value of the joint fingerprinting game is defined by
\begin{equation}\label{eqn:minimax_joint_rand}
\overline{C}_k^\joint(\scrP^\emb, \scrP^\col) = \min_{p_{Y|X_\sfK} \in \scrP^\col} ~ \max_{P_\bfW \in \scrP^\emb} ~ \bbE_{P_\bfW} \left[ I^\joint_k(\bfW, p_{Y|X_\sfK}) \right].
\end{equation}
\end{definition}

\begin{definition}
The minimax value of the simple fingerprinting game is defined by
\begin{equation}\label{eqn:minimax_simple_rand}
\overline{C}_k^\simple(\scrP^\emb, \scrP^\col) = \min_{p_{Y|X_\sfK} \in \scrP^\col_\fair} ~ \max_{P_\bfW \in \scrP^\emb} ~ \bbE_{P_\bfW} \left[ I^\simple_k(\bfW, p_{Y|X_\sfK}) \right].
\end{equation}
\end{definition}

When Condition \ref{cond:prob} is satisfied (for example, when $\scrP^\emb = \scrP_\bfW$), the minimax games can be simplified as the following:

\begin{lemma}\label{lem:minimax}
Assume that Condition \ref{cond:prob} is satisfied. Then the minimax values of the joint and simple fingerprinting games can be respectively written as
\begin{equation}\label{eqn:minimax_joint}
\overline{C}_k^\joint(\scrP^\emb, \scrP^\col) = \min_{p_{Y|X_\sfK} \in \scrP^\col} ~ \max_{\bfw \in \supp(\scrP^\emb)} ~ I^\joint_k(\bfw, p_{Y|X_\sfK})
\end{equation}
and
\begin{equation}\label{eqn:minimax_simple}
\overline{C}_k^\simple(\scrP^\emb, \scrP^\col) = \min_{p_{Y|X_\sfK} \in \scrP^\col_\fair} ~ \max_{\bfw \in \supp(\scrP^\emb)} ~ I^\simple_k(\bfw, p_{Y|X_\sfK}).
\end{equation}
\end{lemma}

\begin{proof}
Note that randomization is no longer necessary for the minimax games of (\ref{eqn:minimax_joint_rand}) and (\ref{eqn:minimax_simple_rand}), so they have the same respective values as (\ref{eqn:minimax_joint}) and (\ref{eqn:minimax_simple}).
\end{proof}

We now present the main saddle-point property of the fingerprinting games. The results owe to the convexity of the payoff function with respect to the minimizer's strategy. Such games are generally called \removed{the class of} convex games \cite[$\S2.5$]{Petrosjan1996}.

\begin{theorem}\label{thm:minimax}
For compact $\scrP^\emb$, compact and user-symmetric $\scrP^\col$, and for both the joint and simple games, \changed{$\overline{C}_k(\scrP^\emb, \scrP^\col) = C_k(\scrP^\emb, \scrP^\col)$}\removed{$\overline{C}_k(\scrP^\emb, \scrP^\col) = \underline{C}_k(\scrP^\emb, \scrP^\col)$}. Suppose further that $\scrP^\emb$ and $\scrP^\col$ are symbol-symmetric, then the first argument $\scrP^\emb$ can be replaced by $\scrP^\emb_\sym$ and/or the second argument $\scrP^\col$ can be replaced by $\scrP^\col_\fair$, $\scrP^\col_\sym$, or $\scrP^\col_{\fair,\sym}$ without changing the minimax or the maximin value. For all these games, the minimizer has an optimal strategy $p_{Y|X_\sfK}^{(k)} \in \scrP^\col_{\fair,\sym}$ while the maximizer has an optimal strategy $P_\bfW^{(k)} \in \scrP^\emb_\sym$. In particular, when Condition \ref{cond:prob} is satisfied, the maximizing strategy $p_\bfW^{(k)} \in \scrP^\emb_\sym$ has a finite spectrum. The values of all these games equal the (joint or simple) fingerprinting capacity $C_k(\scrP^\emb, \scrP^\col)$.
\end{theorem}
\begin{proof}
We show that the functions $I^\joint_k$ and $I^\simple_k$ are convex functions of $p_{Y|X_\sfK}$ for fixed $\bfw$. The convexity of $I^\joint_k$ is shown in \cite[Theorem 2.7.4]{Cover2006}. To show the convexity of $I^\simple_k(\bfw,p_{Y|X_\sfK})$, we fix $\bfw$ and consider two different conditional distributions $p_{Y|X_\sfK}^1$ and $p_{Y|X_\sfK}^2$. Note that
\begin{eqnarray}
I^\simple_k(\bfw,p_{Y|X_\sfK}) &=& I(X_1;Y|\bfW = \bfw)\nonumber\\
&=& \sum_{x,y} p_{X_1|\bfW}(x|\bfw) p_{Y|X_1 \bfW}(y|x,\bfw) \log \frac{p_{Y|X_1 \bfW}(y|x,\bfw)}{p_{Y|\bfW}(y|\bfw)}\nonumber\\
&=& \sum_{x} p_{X_1|\bfW}(x|\bfw) D( p_{Y|X_1 \bfW} \parallel p_{Y|\bfW} | \bfW = \bfw).
\end{eqnarray}
For any $0 \leq \lambda \leq 1$, we have
\begin{eqnarray}
\lefteqn{\lambda I^\simple(\bfw,p_{Y|X_\sfK}^1) + (1-\lambda) I^\simple(\bfw,p_{Y|X_\sfK}^2)}\nonumber\\
&=& \sum_{x} p_{X_1|\bfW}(x|\bfw) \left[ \lambda D( p^1_{Y|X_1 \bfW} \parallel p^1_{Y|\bfW} | \bfW = \bfw) + (1-\lambda) D( p^2_{Y|X_1 \bfW} \parallel p^2_{Y|\bfW} | \bfW = \bfw)\right] \nonumber\\
&\geq& \sum_{x} p_{X_1|\bfW}(x|\bfw) D(p^\lambda_{Y|X_1 \bfW} \parallel p^\lambda_{Y|\bfW} | \bfW = \bfw)\nonumber\\
&=& I^\simple(\bfw,p_{Y|X_\sfK}^\lambda)
\end{eqnarray}
where $p_{Y|X_\sfK}^\lambda = \lambda p_{Y|X_\sfK}^1 + (1-\lambda) p_{Y|X_\sfK}^2 \in \scrP^\col$ by compactness. The inequality follows from the convexity of relative entropy \cite[Theorem 2.7.2]{Cover2006}. Hence $I^\simple_k$ is convex in $p_{Y|X_\sfK}$.

Now since $I_k(\bfw,p_{Y|X_\sfK})$ is a convex function of $p_{Y|X_\sfK}$ for fixed $\bfw \in \scrW^q$, $\bbE_{P_\bfW} [I_k(\bfW,p_{Y|X_\sfK})]$ is also a convex function of $p_{Y|X_\sfK}$ for fixed $P_\bfW \in \scrP_\bfW$. On the other hand, $\bbE_{P_\bfW} [I_k(\bfW,p_{Y|X_\sfK})]$ is a linear function of $P_\bfW$ for fixed $p_{Y|X_\sfK}$. By the minimax theorem \cite{Sion1958}, the game admits a saddle-point solution.

If $\scrP^\emb$ is symbol-symmetric and let $P_\bfW$ be a minimizing saddle-point strategy, then by symbol-symmetry each $P_\bfW^\pi \in \scrP^\emb$ is a minimizing saddle-point strategy for any permutation $\pi$ of $\calX$. The symbol-permutation averaged distribution
\begin{equation}
\overline{P}_\bfW = \frac{1}{q!} \sum_\pi P_\bfW^\pi
\end{equation}
is also a minimizing saddle-point strategy and is symbol-symmetric by construction. Similarly if $\scrP^\col$ is user-symmetric and symbol-symmetric, we can construct a maximizing saddle-point strategy that is both user-symmetric and symbol-symmetric.

Finally if $\scrP^\emb$ is the class of all probability distributions on $\supp(\scrP^\emb)$ (Condition \ref{cond:prob}), the game becomes a so-called convex game whose minimizing strategy has a finite spectrum (see \cite[$\S2.5$]{Petrosjan1996}).
\end{proof}
\section{Fingerprinting capacity for the binary alphabet}\label{sec:bin}

In the following two sections we study intensively the joint and simple fingerprinting games for the binary alphabet, i.e. $\calX = \calY = \{0,1\}$.\footnote{In the case of the binary alphabet, the four variations of the marking assumption discussed in \cite{vSkori'c2008} are equivalent in terms of capacity. Hence for simplicity we assume $\calX = \calY$.} Tight upper and lower bounds on capacities are provided under several different setups. \removed{Using asymptotic methods we show the optimality of the arcsine distribution and the interleaving attack for large coalitions.}

\subsection{Game Definition}\label{ssec:bin_def}

The mutual information games \changed{for}\removed{of} joint and simple decoder in the binary case can be simplified as follows:\looseness=-1

\begin{enumerate}
\item {\bf Fingerprinting Codes}

The auxiliary random vector $\bfW$ now has only one degree of freedom, and we redefine it as $W \in [0,1]$. $P_W$ denotes its distribution and $p_{X |W} \sim \Bernoulli(W)$.

Suppose $\scrP^\emb$ is compact and symbol-symmetric. Then by Theorem \ref{thm:minimax}, it suffices to consider symbol-symmetric $P_W$, which in the binary case means that the distribution of $W$ is symmetric about $1/2$, i.e.,\looseness=-1
\begin{equation}\label{eqn:bin_embsym}
Pr(W \leq w) = Pr(W \geq 1-w), \quad w \in [0,1].
\end{equation}
In the numerical results in Sec. \ref{ssec:Beta}, we will consider a subset of the family of beta distributions, which is a family of continuous probability distributions defined on $(0,1)$:
\begin{equation}\label{betapdf}
f^\theta_W(w) = \frac{1}{B(\theta,\theta)} \left[ w(1-w) \right]^{\theta-1}
\end{equation}
where the beta function, $B(\theta,\theta) \triangleq \int_0^1 \left[t(1-t)\right]^{\theta-1} dt$, appears as a normalization constant, and the parameter $\theta > 0$. The arcsine distribution, which is a special case of the beta distribution with $\theta = 1/2$, has pdf
\begin{equation}\label{eqn:arcsinepdf}
f^*_W(w) = \frac{1}{\pi \sqrt{w(1-w)}}
\end{equation}
on $(0,1)$. The arcsine distribution was first used in generating randomized fingerprinting codes by Tardos \cite{Tardos2003} and is sometimes referred to as the ``Tardos distribution'' in the literature.

\item {\bf Collusion Channel}

Suppose $\scrP^\col$ is compact and user- and symbol-symmetric. Then by Theorem \ref{thm:minimax} it suffices to consider user-symmetric attacks. Let $Z \triangleq \sum_{i = 1}^k X_i \in \left\{0, 1, \ldots, k\right\}$, which is the number of 1's in $X_\sfK$. 
User-symmetry makes $Z$ a sufficient statistic in producing $Y$. If we let $\bfp = (p_0, \ldots, p_k)'$ where $p_z \triangleq p_{Y|Z}(1|z), z = 0, \ldots, k$, then the collusion channel can be completely characterized by $\bfp$. The marking assumption enforces that 
\begin{equation}\label{eqn:bin_marking}
p_0 = 0 \textrm{ and } p_k = 1.
\end{equation}
On the other hand, symbol-symmetry allows us to consider $\bfp$ with
\begin{equation}\label{eqn:bin_colsym}
p_z = 1-p_{k-z}, \quad z=0, \ldots, k.
\end{equation}

The \emph{interleaving attack} $\bfp^*$ (a.k.a. ``uniform channel'' in \cite{Anthapadmanabhan2008} and ``blind colluders'' in \cite{Furon2008}) defined by\looseness=-1
\begin{equation}\label{eqn:interleaving}
p^*_z = \frac{z}{k}, \quad z = 0, \ldots, k
\end{equation}
is frequently adopted to model the coalition's \changed{strategy}\removed{behavior} and can be easily implemented by drawing each $y_j$ randomly from $x_{1,j}, \ldots, x_{k,j}$ at each position $j$. One can verify that it satisfies the marking assumption (\ref{eqn:bin_marking}) and is both user- and symbol-symmetric (\ref{eqn:bin_embsym}). We will further discuss the performance of this attack in Sec. \ref{ssec:bin_bounds} and Sec. \ref{sec:bin_asym}.

\item {\bf Payoff Functions}

Let $\bfal(w) = \left(\alpha_0(w), \ldots, \alpha_k(w)\right)'$ and similarily for $\bfal^1(w)$ and $\bfal^0(w)$ where
\begin{equation}\label{eqn:alpha}
\alpha_z(w) \triangleq p_{Z|W}(z|w) = \binom{k}{z} w^z (1-w)^{k-z}
\end{equation}
is the binomial law with parameter $w$ and $k$ trials, and 
\begin{equation}\label{eqn:alpha1}
\alpha^1_z(w) \triangleq p_{Z|X_1 W}(z|1,w) = \left\{ \begin{array}{ll}\binom{k-1}{z-1} w^{z-1} (1-w)^{k-z}, &1 \leq z \leq k\\ 0, &z = 0 \end{array} \right.
\end{equation}
and
\begin{equation}\label{eqn:alpha0}
\alpha^0_z(w) \triangleq p_{Z|X_1 W}(z|0,w) = \left\{ \begin{array}{ll}\binom{k-1}{z} w^{z} (1-w)^{k-z-1}, &0 \leq z \leq k-1\\ 0, &z = k \end{array} \right.
\end{equation}
are the (shifted for $\bfal^1(w)$) binomial laws with parameter $w$ and $k-1$ trials.

Recall that $W \rightarrow X_\sfK \rightarrow Z \rightarrow Y$ forms a Markov chain. We have
\begin{eqnarray}
p_{Y|X_1 W}(1|x,w) &=& \sum_{z=0}^k p_{Z|X_1 W}(z|x,w) p_{Y|Z}(1|z) \nonumber\\
&=& \sum_{z=0}^k \alpha^x_z(w) p_z = {\bfal^{x}}'\bfp \label{eqn:yx1w}
\end{eqnarray}
for $x = 0,1$. The payoff function \changed{for}\removed{of} the joint fingerprinting game is then
\begin{eqnarray}
I^\joint_k(w,\bfp) &=& \frac{1}{k} I(X_\sfK; Y|W = w) \nonumber\\
&=& \frac{1}{k} I(Z; Y|W = w) \nonumber \\
&=& \frac{1}{k} \left[ H(Y|W = w) - H(Y|Z, W = w) \right]\nonumber\\
&=& \frac{1}{k} \left[ h\left(\sum_{z=0}^{k} \alpha_z(w) p_z \right) - \sum_{z=0}^{k} \alpha_z(w) h(p_z) \right] \nonumber\\
&=& \frac{1}{k} \left[h(\bfal'\bfp) - \bfal' h(\bfp)\right].\label{eqn:I_joint_1}
\end{eqnarray}

Another representation of $I^\joint_k$ is
\begin{eqnarray}
I^\joint_k(w,\bfp) &=& \frac{1}{k} D(p_{ZY|W} \parallel p_{Z|W}p_{Y|W} | W = w) \nonumber\\
&=& \frac{1}{k} \sum_{z=0}^{k} \sum_{y=0}^{1} p_{Z|W}(z|w) p_{Y|Z}(y|z) \log \frac{p_{Y|Z}(y|z)}{p_{Y|W}(y|w)} \nonumber\\
&=& \frac{1}{k} \sum_{z=0}^{k} \alpha_z(w) \left[ p_z \log \frac{p_z}{\bfal'\bfp} + (1-p_z) \log \frac{1-p_z}{1-\bfal'\bfp} \right] \nonumber\\
&=& \frac{1}{k} \sum_{z=0}^{k} \alpha_z(w)~d(p_z \parallel \bfal'\bfp).\label{eqn:I_joint_2}
\end{eqnarray}

For the simple fingerprinting game, we have
\begin{eqnarray}\label{eqn:I^simple}
I^\simple_k(w,\bfp) &=& I(X_1; Y|W = w) \nonumber\\
&=& D(p_{X_1 Y|W} \parallel p_{X_1|W}p_{Y|W} | W = w) \nonumber\\
&=& \sum_{x=0}^{1} \sum_{y=0}^{1} p_{X_1|W}(x|w) p_{Y|X_1 W}(y|x,w) \log \frac{p_{Y|X_1 W}(y|x,w)}{p_{Y|W}(y|w)} \nonumber\\
&=& w D(p_{Y|X_1=1, W} \parallel p_{Y|W}) + (1-w) D(p_{Y|X_1=0, W} \parallel p_{Y|W}) \nonumber \\
&\stackrel{(a)}{=}& wd({\bfal^1}'\bfp\parallel\bfal'\bfp)+(1-w)d({\bfal^0}'\bfp\parallel\bfal'\bfp)\label{eqn:I_simple}
\end{eqnarray}
where (a) follows from (\ref{eqn:yx1w}).

\item {\bf Fingerprinting Games}

The fingerprinting games for the binary alphabet under the marking assumption can now be written as
\begin{eqnarray}
C_k(\scrP^\emb,\scrP^\col) &=& \max_{P_W} ~ \min_{\bfp} ~ \bbE_{P_W}\left[I_k(W,\bfp)\right]\label{eqn:bin_maximin}\\
&=& \min_{\bfp} ~ \max_{P_W} ~ \bbE_{P_W}\left[I_k(W,\bfp)\right]\label{eqn:bin_minimax_rand}
\end{eqnarray}
where the maximization is subject to $P_W \in \scrP^\emb$ and the symbol-symmetry condition (\ref{eqn:bin_embsym}) while the minimization is subject to $\bfp \in \scrP^\col$ and the symbol-symmetry condition (\ref{eqn:bin_colsym}). The maximizing and minimizing strategies are denoted by $P_W^{(k)}$ and $\bfp^{(k)}$ respectively. If $\scrP^\emb$ satisfies Condition \ref{cond:prob}, then by Lemma \ref{lem:minimax} we have
\begin{equation}\label{eqn:bin_minimax}
C_k(\scrP^\emb,\scrP^\col) = \min_{\bfp} ~ \max_{w} ~ I_k(w,\bfp)
\end{equation}
where the maximization is subject to $w \in \supp(\scrP^\emb)$.

\end{enumerate}

\subsection{Analysis of the Convex Games}\label{ssec:bin_analysis}

We consider the following three cases:

\begin{enumerate}
\item {\bf Colluders' Strategy is Fixed}

$\scrP^\col = \{ \bfp \}$ in this case. For general $\scrP^\emb$ the game is still an infinite-dimensional maximization problem. However when Condition \ref{cond:prob} is satisfied, it reduces to one-dimensional by (\ref{eqn:bin_minimax}) and a simple line search gives us the capacity under collusion channel $\bfp$. Note that for any $\bfp \in \scrP_\marking$, $C_k(\scrP^\emb,\bfp)$ is an upper bound on $C_k(\scrP^\emb,\scrP_\marking)$.

\item {\bf Fingerprinting Embedder's Strategy is Fixed}

$\scrP^\emb = \{ P_W \}$ in this case. The game reduces to a $\Theta(k)$-dimentional minimization problem. Since the payoff function $\bbE_{P_W}\left[I_k(W,\bfp)\right]$ is convex in $\bfp$, we use the conditional gradient method to solve the constrained convex optimization problem (\ref{eqn:bin_maximin}) (see \cite{Bertsekas1999}). For the joint fingerprinting game, Furon and Perez-Freire \cite{Furon2009a} proposed a Blahut-Arimoto algorithm which, however, cannot be applied to the simple fingerprinting game. Note that for any $P_W \in \scrP_W$, $C_k(P_W,\scrP^\emb)$ is a lower bound on $C_k(\scrP_W,\scrP^\emb)$.\looseness=-1

\item {\bf Fingerprinting Capacities Under the Marking Assumption}

We consider specifically $\scrP^\emb = \scrP_W$ and $\scrP^\col = \scrP_\marking$. Solving the maximin game of (\ref{eqn:bin_maximin}) or the minimax game of (\ref{eqn:bin_minimax}) is much more difficult than solving the above maximization or minimization problems. In particular, the alternating maximization and minimization algorithm generally diverges.

Owing to the existence of a saddle-point solution, $\bfp^{(k)}$ and $p^{(k)}_W$ (note that it is a pmf by Theorem \ref{thm:minimax}) must satisfy the following:
\begin{enumerate}
\item When $\bfp = \bfp^{(k)}$ is fixed, $I(w,\bfp^{(k)})$ is a differentiable function over the unit interval. The support $\supp\left(p_W^{(k)}\right)$ of $p_W^{(k)}$ can only take values at the maximizers of $I(w,\bfp^{(k)})$ \cite[$\S2.5$]{Petrosjan1996}. Hence we have
\begin{equation} \label{eqn:dIdw}
\left\{\begin{array}{l}
I(w,\bfp^{(k)}) = C_k(\scrP_W,\scrP_\marking)\\
\frac{\partial}{\partial w} I(w,\bfp^{(k)}) = 0\\
\frac{\partial^2}{\partial w^2} I(w,\bfp^{(k)}) < 0
\end{array}\right., \quad \forall w \in \supp\left(p_W^{(k)}\right).
\end{equation}

\item When $p_W = p^{(k)}_W$ is fixed and the constraint (\ref{eqn:bin_colsym}) is imposed, we have
\begin{equation} \label{eqn:dIdp}
\bbE_{p^{(k)}_W}\left[\frac{\partial}{\partial
p_z}I(W,\bfp^{(k)})\right] = 0, \quad z = 1,\ldots, \left\lfloor
\frac{k-1}{2} \right\rfloor.
\end{equation}
\end{enumerate}

By the convexity in $\bfp$ of the payoff function, we have $\left| \supp\left(p_W^{(k)}\right) \right| \leq \left\lfloor\frac{k+1}{2}\right\rfloor$ (see \cite[$\S2.5$]{Petrosjan1996}). With a fixed spectrum cardinality, we can obtain candidate capacity-achieving strategies $p^{(k)}_W$ and $\bfp^{(k)}$ by solving (\ref{eqn:dIdw}) and (\ref{eqn:dIdp}), and
then verify whether those candidate distributions are optimal by examining
the second partial derivatives. Once $\bfp^{(k)}$ and $p^{(k)}_W$ are
found, we can evaluate $C_k(\scrP_W,\scrP_\marking)$ by substituting them into (\ref{eqn:bin_maximin}).

\end{enumerate}

Numerical solutions to the joint and simple fingerprinting games are shown in Fig. \ref{fig:C_plot}-\ref{fig:p_plot_TS_mark}. \added{Observe that capacities for both games (Fig. \ref{fig:C_plot}(a)-(b)), the optimal distributions ${P_W^{(k)}}$ for both games (Fig. \ref{fig:cdf_plot}(a)-(b)), and the optimal attacks $\bfp^{(k)}$ for the joint fingerprinting game (Fig. \ref{fig:p_plot_TS_mark}(a)) all seem to converge as $k$ grows. The optimal attacks $\bfp^{(k)}$ for the simple fingerprinting game (Fig. \ref{fig:p_plot_TS_mark}(b)) however exhibit some wild oscillations in both amplitude and frequency as $k$ grows. We will study the asymptotics of the joint fingerprinting game for large $k$ in the next section.}

\begin{figure*}[t]
\begin{minipage}{.48\textwidth}
  \centering
  \includegraphics[width=.85\textwidth]{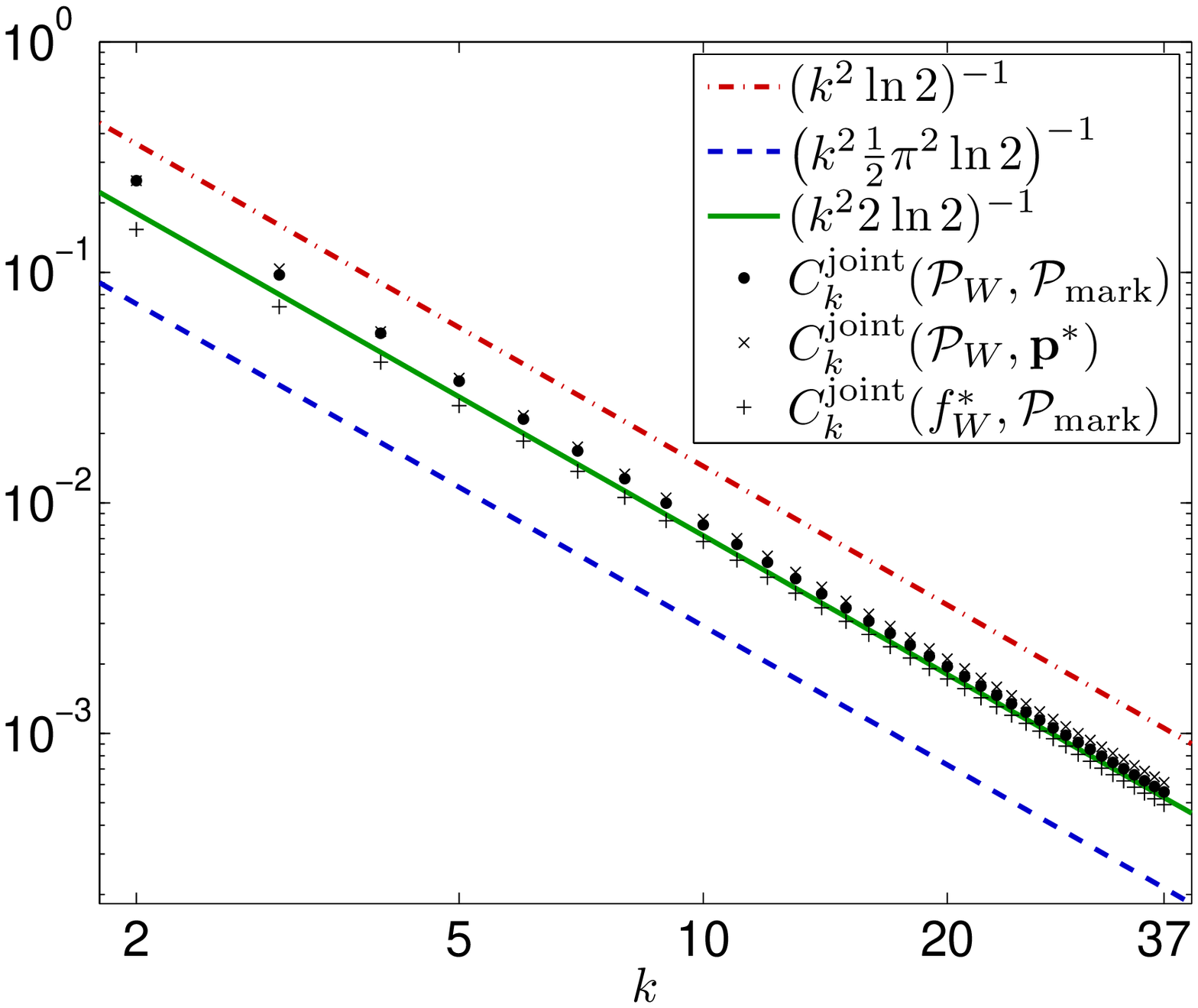}
  \centerline{\added{(a) Joint}}
\end{minipage}
\hfill
\begin{minipage}{.48\textwidth}
  \centering
  \includegraphics[width=.85\textwidth]{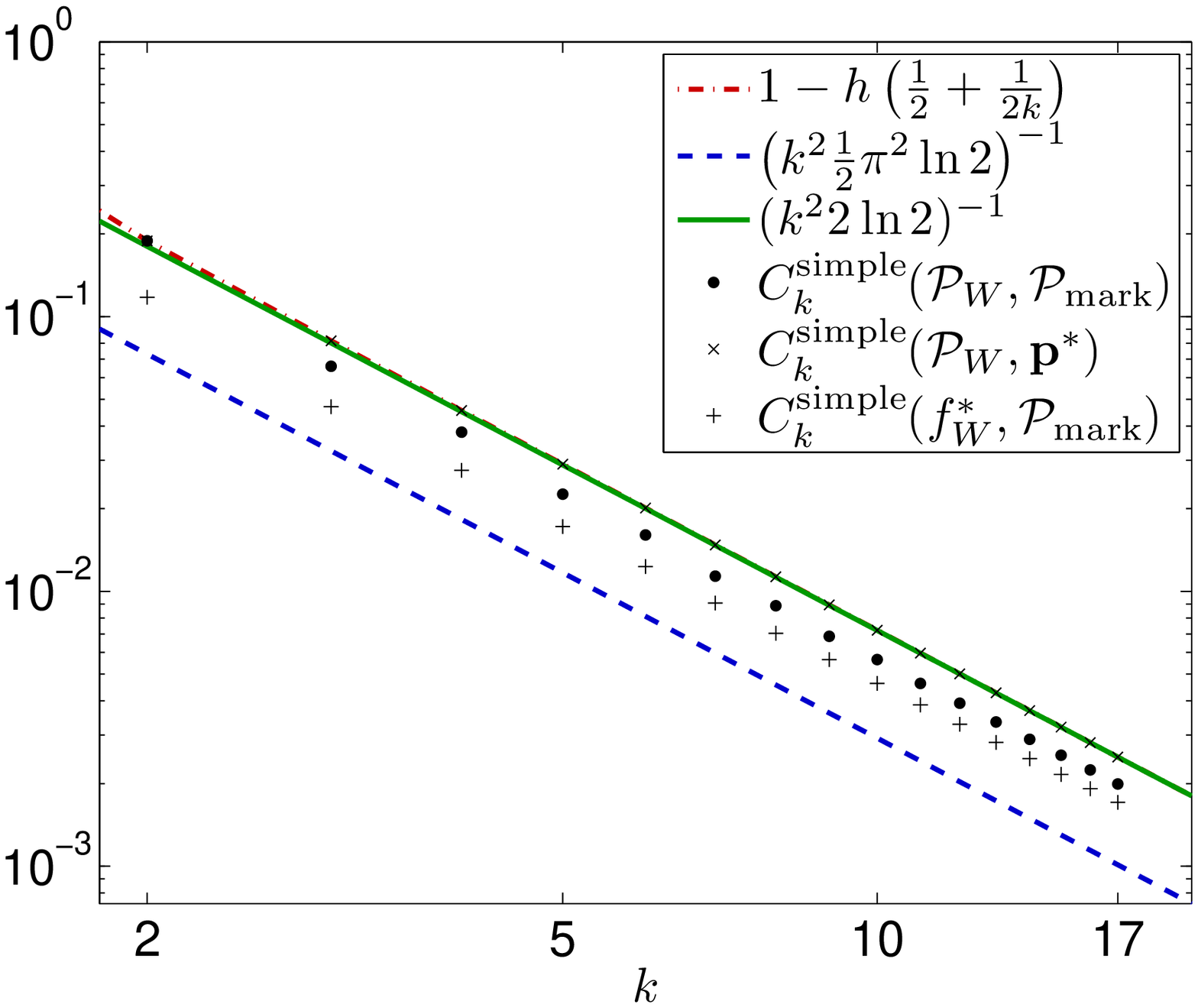}
  \centerline{\added{(b) Simple}}
\end{minipage}

\caption{Capacities $C_k(\scrP_W,\scrP_\marking)$, $C_k(\scrP_W,\bfp^*)$, $C_k(f^*_W,\scrP_\marking)$, and upper and lower bounds of the \added{(a)} joint and \added{(b)} simple fingerprinting games}
\label{fig:C_plot}
\end{figure*}

\begin{figure*}[t]
\begin{minipage}{.48\textwidth}
  \centering
  \includegraphics[width=.85\textwidth]{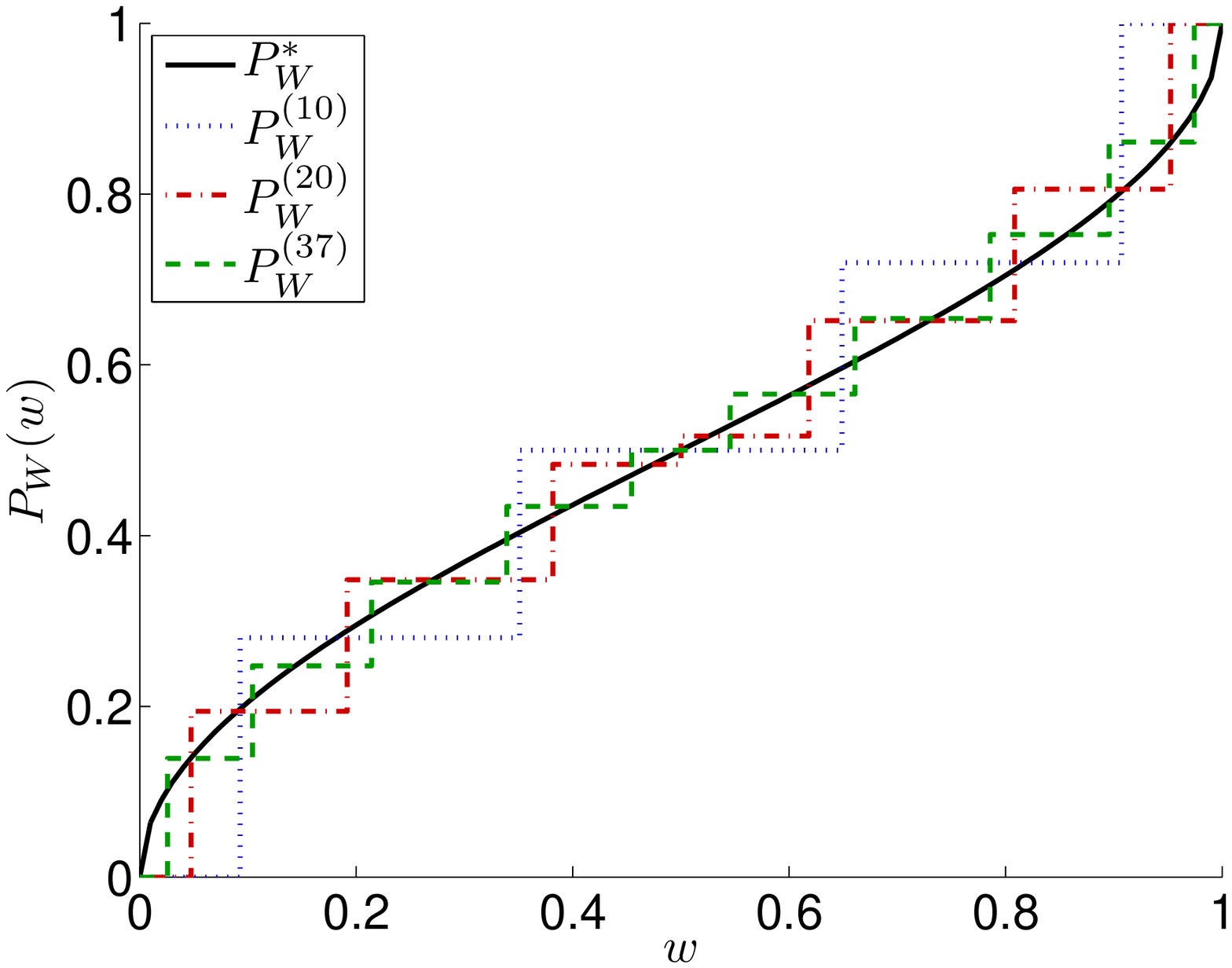}
  \centerline{\added{(a) Joint}}
\end{minipage}
\hfill
\begin{minipage}{.48\textwidth}
  \centering
  \includegraphics[width=.85\textwidth]{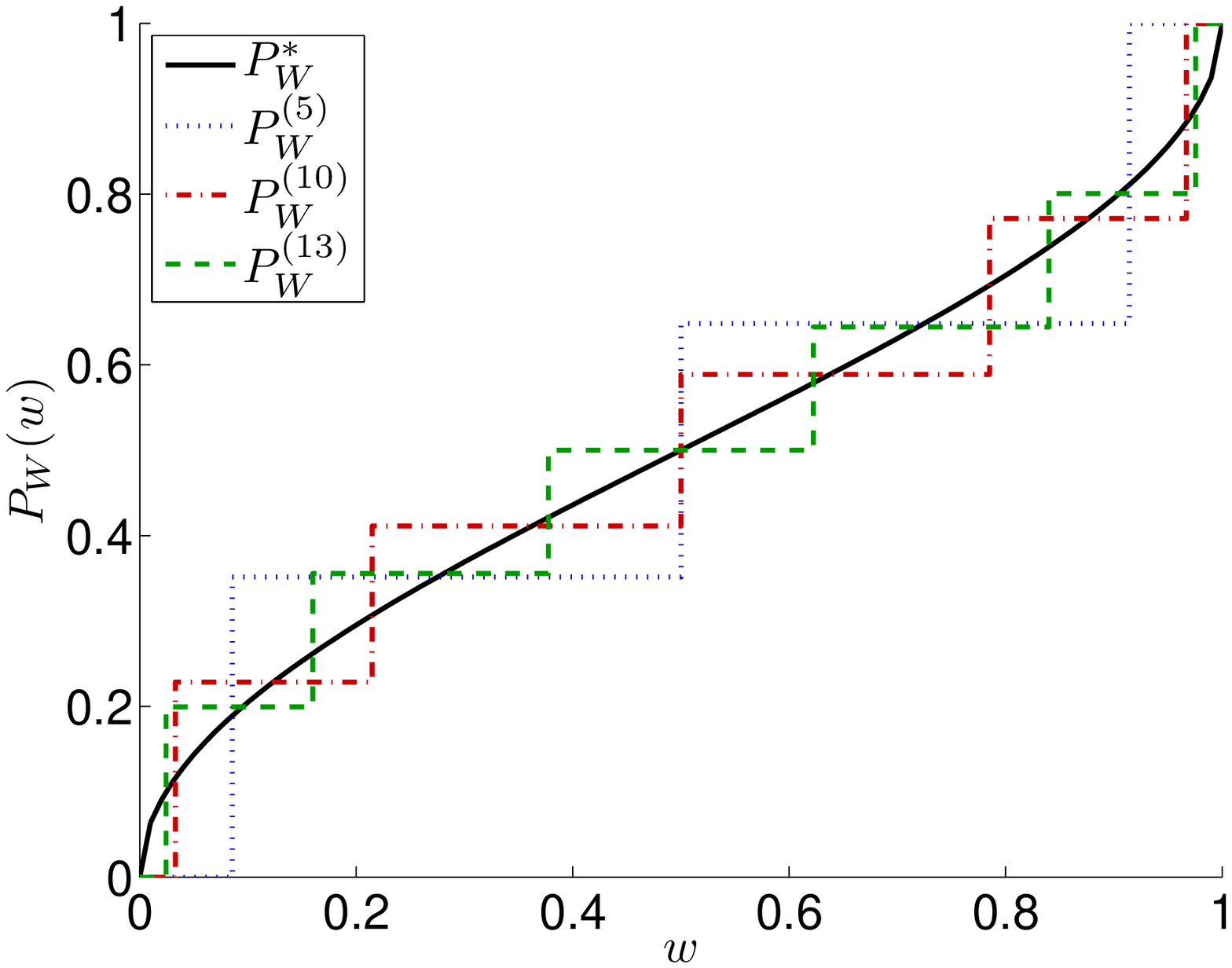}
  \centerline{\added{(b) Simple}}
\end{minipage}

\caption{Cumulative distribution function of the arcsine distribution $P_W^*$ and \added{the optimal distributions} ${P_W^{(k)}}$ \changed{for}\removed{of} the \added{(a)} joint and \added{(b)} simple fingerprinting games under the marking assumption}
\label{fig:cdf_plot}
\end{figure*}

\begin{figure*}[t]
\begin{minipage}{.48\textwidth}
  \centering
  \includegraphics[width=.85\textwidth]{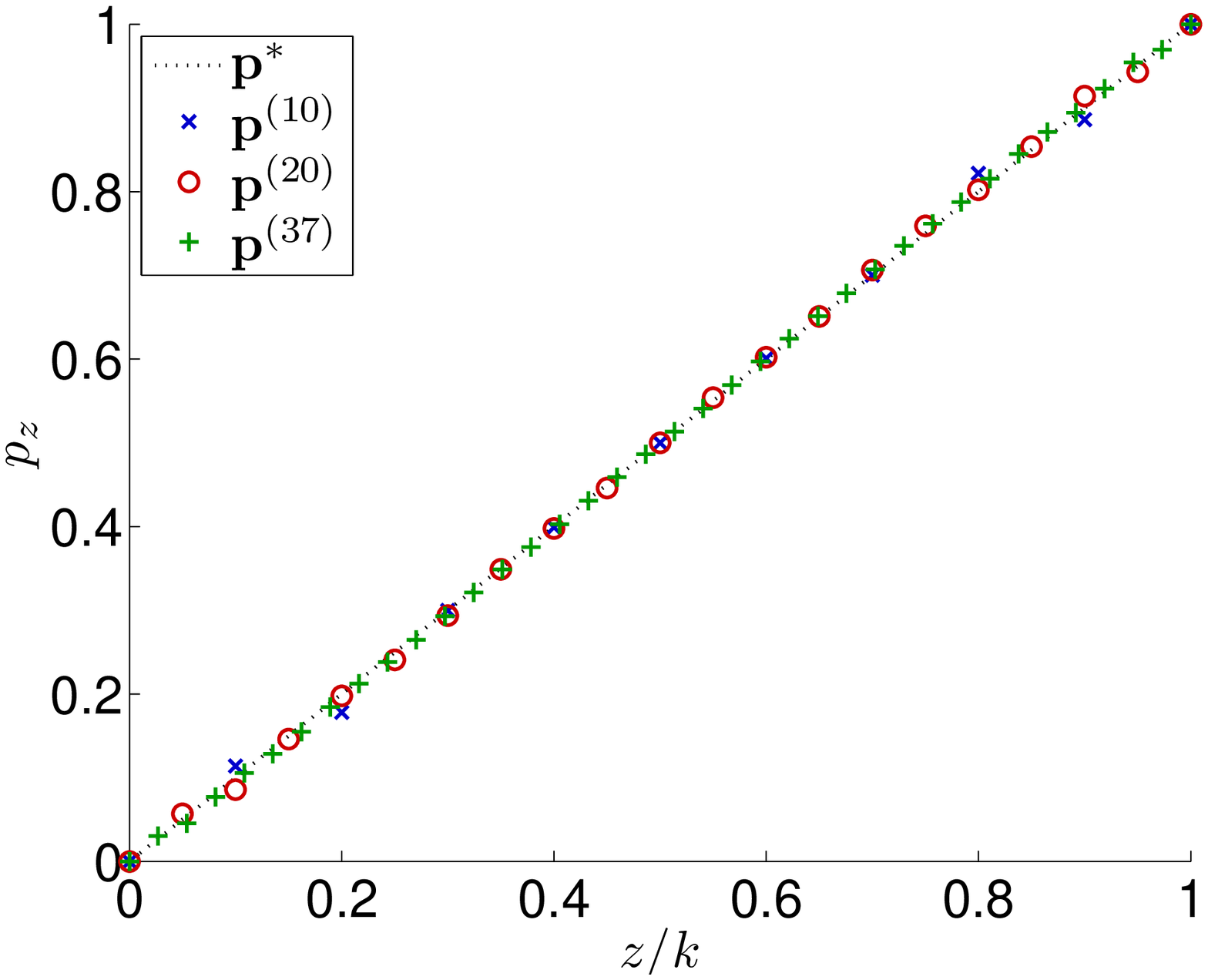}
  \centerline{\added{(a) Joint}}
\end{minipage}
\hfill
\begin{minipage}{.48\textwidth}
  \centering
  \includegraphics[width=.85\textwidth]{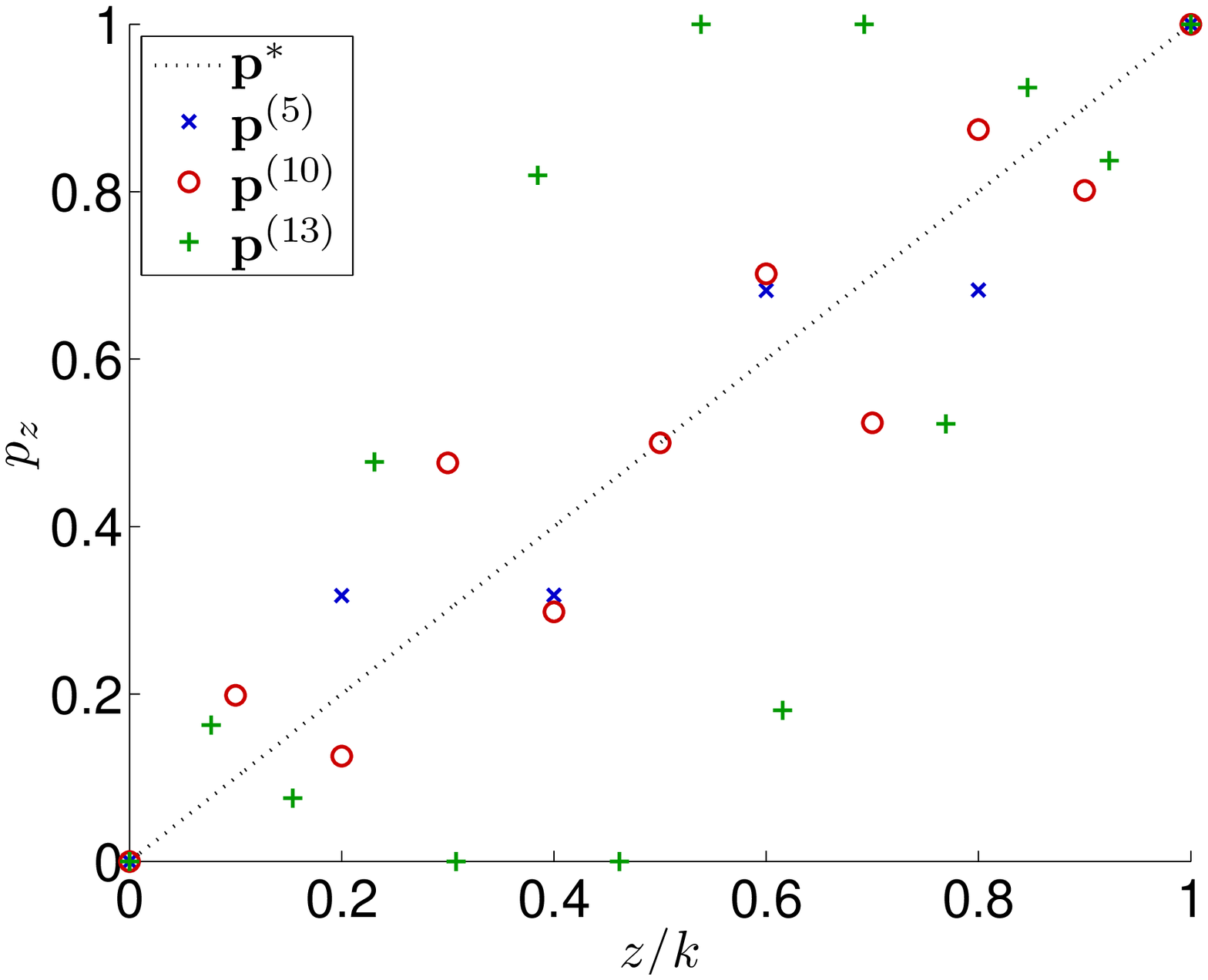}
  \centerline{\added{(b) Simple}}
\end{minipage}

\caption{The interleaving attack $\bfp^*$ and \added{the optimal attacks} ${\bfp^{(k)}}$ \changed{for}\removed{of} the \added{(a)} joint and \added{(b)} simple fingerprinting games under the marking assumption}
\label{fig:p_plot_TS_mark}
\end{figure*}

\subsection{Capacity Bounds}\label{ssec:bin_bounds}

The analysis of Sec. \ref{ssec:bin_analysis} allows us to solve the fingerprinting game numerically for small $k$. However, evaluating or even approximating the capacity value for large $k$ is still a difficult task. In this subsection, we provide tight upper and lower bounds on capacity.

For simplicity of notation, we let
\begin{equation}\label{eqn:g_k}
g_k(w) \triangleq p_{Y|W}(1|w) = \sum_{z=0}^{k} \alpha_z(w) p_z = \bfal(w)'\bfp
\end{equation}
which by the definition of $\alpha_z(w)$ in (\ref{eqn:alpha}) is a polynomial in $w$ of degree $\leq k$. Note that
$g_k(0) = p_0 = 0$ and $g_k(1) = p_k = 1$ by the marking assumption.

The following lemmas will be useful for the proofs:

\begin{lemma}[Pinsker's inequality]\cite[Lemma 11.6.1]{Cover2006}\label{lem:pinsker}
\begin{equation}
d(r \parallel s) \geq \frac{2}{\ln 2}(r-s)^2.
\end{equation}
\end{lemma}

\begin{lemma}\label{lem:der_iden}
Equalities
$${\bfal^1}'\bfp -\bfal'\bfp = \frac{1-w}{k} g_k'(w)$$
and
$${\bfal^0}'\bfp -\bfal'\bfp = -\frac{w}{k} g_k'(w)$$
hold for $z = 0, \ldots, k$ and $w \in [0,1]$.
\end{lemma}
\begin{proof}
The equalities follow directly from (\ref{eqn:alpha}), (\ref{eqn:alpha1}), (\ref{eqn:alpha0}), and (\ref{eqn:g_k}).
\end{proof}

\begin{lemma}\label{lem:lb_f_int}
Let $f_W$ be a pdf on $[0,1]$. Then
\begin{equation}
\int_0^1 \frac{dw}{f_W(w)w(1-w)} \geq \pi^2
\end{equation}
with equality if and only if $f_W$ is the arcsine distribution of (\ref{eqn:arcsinepdf}).
\end{lemma}
\begin{proof}
By the Cauchy-Schwarz inequality, we have
\begin{equation*}
\int_0^1 \frac{dw}{f_W(w)w(1-w)} \geq \frac{\left[ \int_0^1 \frac{dw}{\sqrt{w(1-w)}} \right]^2}{\int_0^1 f_W(w) dw} = \pi^2.
\end{equation*}
Equality holds if and only if $f_W(w) \propto \frac{1}{\sqrt{w(1-w)}}$, which leads us to the arcsine distribution.
\end{proof}

\subsubsection{Upper Bounds}

The following two theorems bound from above capacities under the interleaving attack of (\ref{eqn:interleaving}).

\begin{theorem}\cite[Theorem 4.2]{Huang2009}\label{thm:ub_joint}
\begin{equation}
C_k^\joint(\scrP_W,\bfp^*) \leq \frac{1}{k^2 \ln 2}.
\end{equation}
\end{theorem}

\begin{proof}
\begin{eqnarray*}
C_k^\joint(\scrP_W,\bfp^*) &=& \max_{w \in [0,1]}~I^\joint_k(w,\bfp^*)\nonumber\\
&\stackrel{(a)}{=}& \frac{1}{k} \max_{w \in [0,1]}~ \left\{h(w)-\sum_{z=0}^{k}\alpha_z(w)h\left(\frac{z}{k}\right)\right\}\nonumber\\
&\stackrel{(b)}{\leq}& \frac{1}{k^2 \ln 2}
\end{eqnarray*}
where (a) follows from (\ref{eqn:I_joint_1}) and (b) results from \cite[Theorem 4.3]{Anthapadmanabhan2008}.
\end{proof}

\begin{theorem}\cite[Proposition 4.2]{Huang2009a}\label{thm:ub_simple}
\begin{equation}
C_k^\simple(\scrP_W,\bfp^*) = 1-h\left(\frac{1}{2}+\frac{1}{2k}\right) = \frac{1}{k^2 2\ln 2}+O\left(\frac{1}{k^4}\right).
\end{equation}
\end{theorem}

\begin{proof}
It can be shown that $I^\simple_k(w,\bfp^*)$ takes its maximum at $w = 1/2$ (See the Appendix). Hence
\begin{eqnarray*}
C_k^\simple(\scrP_W,\bfp^*) &=& \max_{w \in [0,1]} ~ I^\simple_k\left(w,\bfp^* \right)\\
&=& 1-h\left(\frac{1}{2}+\frac{1}{2k}\right) = \frac{1}{k^2 2 \ln 2} + O\left(\frac{1}{k^4}\right).
\end{eqnarray*}
\end{proof}

\subsubsection{Lower Bounds}

The following theorem provides a lower bound on both the joint and simple capacities under a continuous probability distribution $f_W$.

\begin{theorem}\label{thm:lb_general}
Let $f_W$ be the pdf of a continuous probability distribution on $[0,1]$. Then
\begin{equation}
C^\joint_k(f_W,\scrP_\marking) \geq C^\simple_k(f_W,\scrP_\marking) \geq \frac{2}{k^2 \ln2} \left[ \int_0^1 \frac{dw}{f_W(w)w(1-w)} \right]^{-1}.
\end{equation}
The lower bound is maximized when $f_W = f_W^*$ where it takes the value $\dfrac{2}{k^2 \pi^2 \ln 2}$.
\end{theorem}
\begin{proof}
For any $\bfp \in \scrP_\marking$, we have
\begin{eqnarray*}
\bbE_{f_W}\left[I^\simple_k(W,\bfp)\right] &\stackrel{(a)}{=}& \int_0^1 \left[ wd({\bfal^1}'\bfp\parallel\bfal'\bfp) + (1-w)d({\bfal^0}'\bfp \parallel \bfal'\bfp) \right] f_W(w)dw\\
&\stackrel{(b)}{\geq}& \frac{2}{\ln 2} \int_0^1 \left[ w({\bfal^1}'\bfp-\bfal'\bfp)^2 + (1-w)({\bfal^0}'\bfp-\bfal'\bfp)^2 \right]f_W(w)dw\\
&\stackrel{(c)}{=}& \frac{2}{k^2 \ln 2} \int_0^1 \left[ g_k'(w) \right]^2 w(1-w) f_W(w)dw\\
&\stackrel{(d)}{\geq}& \frac{2}{k^2 \ln 2} \frac{\left[\int_0^1 g_k'(w) dw\right]^2}{\int_0^1 \frac{dw}{f_W(w)w(1-w)}}\\
&\stackrel{(e)}{=}& \frac{2}{k^2 \ln2} \left[ \int_0^1 \frac{dw}{f_W(w)w(1-w)} \right]^{-1}.
\end{eqnarray*}
(a) follows from (\ref{eqn:I^simple}). (b) follows from Pinsker's inequality (Lemma \ref{lem:pinsker}). (c) follows from Lemma \ref{lem:der_iden}. (d) follows from the Cauchy-Schwarz inequality. Finally, (e) follows from the marking assumption. Hence
\begin{eqnarray*}
C^\simple_k(f_W,\scrP_\marking) &=& \min_{\bfp \in \scrP_\marking} \bbE_{f_W}\left[I^\simple_k(W,\bfp)\right]\\
&\geq& \frac{2}{k^2 \ln2} \left[ \int_0^1 \frac{dw}{f_W(w)w(1-w)} \right]^{-1}.
\end{eqnarray*}
Following Lemma \ref{lem:lb_f_int}, the lower bound is maximized when $f_W = f_W^*$, which coincides with the lower bound given in \cite{Huang2009a}.
\end{proof}

The following corollaries summarizes the upper and lower bounds on capacities under the marking assumption:
\begin{corollary}
\begin{equation}
\frac{2}{k^2 \pi^2 \ln 2} \leq C^\joint_k(\scrP_W,\scrP_\marking) \leq \frac{1}{k^2 \ln 2}.
\end{equation}
\end{corollary}

\begin{corollary}\label{cor:simplebd}
\begin{equation}
\frac{2}{k^2 \pi^2 \ln 2} \leq C^\simple_k(\scrP_W,\scrP_\marking) \leq \frac{1}{k^2 2\ln 2}+O\left(\frac{1}{k^4}\right).
\end{equation}
\end{corollary}   
\section{Asymptotic\changed{s}\removed{ Behavior} for Large Coalitions}\label{sec:bin_asym}

The upper and lower bounds on $C^\joint_k(\scrP_W,\scrP_\marking)$ provided in the previous section are within a factor of about five. As can be seen in Fig. \ref{fig:C_plot}, the numerical results suggest that $C^\joint_k(\scrP_W,\scrP_\marking)$ approximates $(k^2 2 \ln2)^{-1}$ even for small values of $k$. Amiri and Tardos \cite{Amiri2009} claimed the same asymptotic rate \changed{but only provided the proof for the lower bound in \cite[Theorem 15]{Amiri2010}}\removed{without providing a proof}. Here we analyze not only this rate but the complete asymptotic\changed{s}\removed{ behavior} of the joint fingerprinting game.

\subsection{Aymptotic Analysis}\label{ssec:asym_anal}

We consider the sequence of mutual information games \changed{for}\removed{of} joint decoding. To study the asymptotics when $k \rightarrow \infty$, we first assume that the collusion channel $\bfp$ satisfies the following regularity condition:

\begin{condition}\label{cond:cont}
There exists a bounded and twice differentiable function $g(w)$ on $[0,1]$ with $g(0) = 0$ and $g(1) = 1$ such that
\begin{equation}\label{eqn:g_restrict}
p_z = g\left(\dfrac{z}{k}\right), \quad \forall z \in \{0, \ldots, k\}.
\end{equation}
\end{condition}

Certainly the condition restricts the colluders' strategy to a smaller space. We however claim that this is a very mild limitation on their power for the following reasons:
\begin{enumerate}
\item For each $k$, the collusion channels take values of $g$ at only the lattice points in $[0,1]$, hence intuitively the \added{class of collusion channels}\removed{space of $g$} satisfying Condition \ref{cond:cont} remains large.
\item Fig. \ref{fig:p_plot_Beta_mark} shows the minimizing collusion channels $\bfp^{(k)}$ for several different embedding distributions. For each case it seems the continuous interpolation of $\bfp$ does converge to some $g$ on $[0,1]$. Indeed, our following analysis still holds if we relax the restriction of (\ref{eqn:g_restrict}) to
\begin{equation}
p_z = g\left( \dfrac{z}{k}\right) + o\left(\dfrac{1}{k}\right), \quad \forall z \in \{0, \ldots, k\}.
\end{equation}
\end{enumerate}

The following reparameterization of the class (\ref{eqn:g_restrict}) of collusion channels will simplify our analysis:
\begin{definition}\label{def:G_J}
Let $G$ and $J$ be functions on $[0,1]$ defined as
\begin{equation}\label{eqn:G}
G(w) \triangleq \cos^{-1} [1-2g(w)]
\end{equation}
and
\begin{equation}\label{eqn:J}
J(w) \triangleq w(1-w) [G'(w)]^2
\end{equation}
where $g(w)$ satisfies Condition \ref{cond:cont}.
\end{definition}

The outline of our asymptotic analysis is as follows: we fix $w \in (0,1)$ and we study the asymptotic\changed{s}\removed{ behavior} of $I^\joint_k(w,\bfp)$. The binomial distribution of $Z$ can be approximated by the Gaussian distribution with mean $kw$ and variance $kw(1-w)$, and by which we can approximate the dominating terms of $I^\joint_k(w,\bfp)$. Theorem \ref{thm:I_asym} yields $I^\joint_k(w,\bfp) \sim J(w)/(k^2 2 \ln 2) $, where $J$ is the transformation of $g$ defined in (\ref{eqn:J}). The maximin game with $J$ as the payoff function can be solved explicitly and hence the asymptotic\changed{s}\removed{ behavior} of the fingerprinting game can be obtained.

The following lemmas will be useful for our analysis:

\begin{lemma}\cite[Sec. 2.5]{Bavaud2009}\label{lem:chisqr}
\begin{equation}
d(r \parallel s) = \frac{(r-s)^2}{s(1-s) 2 \ln 2} + O(|r-s|^3).
\end{equation}
\end{lemma}

\begin{lemma}\label{lem:chernoff}
For $Z \sim \Binomial(k,w)$, we have
\begin{equation}
Pr[|Z - kw| \geq \sqrt{k\ln k}] \leq 1/k^2.
\end{equation}
\end{lemma}
\begin{proof}
This is special case of Hoeffding's inequality (see \cite{Hoeffding1963}).
\end{proof}

Recall that the expectation of $Y$ given $W = w$, which we denote by $g_k(w)$ in (\ref{eqn:g_k}), can be written as
$$g_k(w) \triangleq p_{Y|W}(1|w) = \sum_{z=0}^{k} \alpha_z(w) p_z = \sum_{z = 0}^k \alpha_z(w) g\left(\dfrac{z}{k}\right)$$
where
$$\alpha_z(w) \triangleq p_{Z|W}(z|w) = \binom{k}{z} w^z (1-w)^{k-z}$$
is the binomial pmf which concentrates around its mean $kw$ as $k \rightarrow \infty$. $g_k$ is a polynomial in $w$ of degree $\leq k$ and is known as the Bernstein polynomial of order $k$ of the function $g$ \cite{Lorentz1986}. By Condition \ref{cond:cont} $g$ is bounded and the second derivative $g''(w)$ exists, from Bernstein \cite[\S1.6]{Lorentz1986} we have
\begin{equation}
g_k(w) = g(w) + \frac{w(1-w)}{2k} g''(w) + o\left(\frac{1}{k}\right).
\label{eqn:asym_gk}
\end{equation}

On the other hand, for $z = k(w+\epsilon)$, we have
\begin{eqnarray}
p_z &=& g\left(\dfrac{z}{k}\right) = g(w+\epsilon) \nonumber\\
&=& g(w) + \epsilon g'(w) + O\left(\epsilon^2\right).
\label{eqn:asym_pz}
\end{eqnarray}

We now write the asymptotic approximation of $I^\joint_k$ in terms of $g$.

Firstly by the bounds presented in Sec. \ref{ssec:bin_bounds}, we can focus on $w$ and $g$ such that $I^\joint_k(w, \bfp) = \Omega(1/k^2)$. Let $\delta = \sqrt{\ln k/k}$. We have
\begin{eqnarray}
I^\joint_k(w,\bfp) &\stackrel{(a)}{=}& \frac{1}{k} \sum_{z=0}^{k} \alpha_z(w)~d(p_z \parallel g_k(w)) \nonumber\\
&\stackrel{(b)}{\sim}& \frac{1}{k} \sum_{z: |z-kw| < k\delta} \alpha_z(w) ~ d(p_z \parallel g_k(w))
\label{eqn:KLdis}
\end{eqnarray}
where (a) follows from (\ref{eqn:I_joint_2}) and (\ref{eqn:g_k}) and (b) from Lemma \ref{lem:chernoff}. 

Now if we let $z = kw + \eta$, where $\eta = O(\sqrt{k\ln k})$, then by (\ref{eqn:asym_pz}) we have
\begin{equation}
p_z = g(w) + \frac{\eta}{k} g'(w) + O\left(\frac{\eta^2}{k^2}\right)
\end{equation}
and combining with (\ref{eqn:asym_gk}) we have
\begin{equation}
p_z - g_k(w) = \frac{\eta}{k} g'(w) + O\left(\frac{\ln k}{k}\right).
\label{eqn:qk_fk_dif}
\end{equation}
for $\eta = \omega(1)$. The contribution to (\ref{eqn:KLdis}) for $\eta = O(1)$ decays exponentially with $k$ and thus can be neglected.

By (\ref{eqn:qk_fk_dif}) and Lemma \ref{lem:chisqr}, we have
\begin{equation}\label{eqn:chi_approx}
d(p_z \parallel g_k(w)) = \frac{[p_z-g_k(w)]^2}{g_k(w)(1-g_k(w)) 2 \ln 2} + o\left(\frac{1}{k}\right)
\end{equation}
and hence (\ref{eqn:KLdis}) yields
\begin{eqnarray}
I^\joint_k(w,\bfp) &\stackrel{(a)}{\sim}& \frac{1}{k2\ln 2}\sum_{z: |z-kw| < k\delta} \alpha_z(w) ~ \frac{[p_z-g_k(w)]^2}{g_k(w)(1-g_k(w))} \nonumber \\
&\stackrel{(b)}{\sim}& \frac{[g'(w)]^2}{k^3 g(w)(1-g(w)) 2\ln 2} \sum_{z: |z-kw| < k\delta} \alpha_z(w) ~ (z-kw)^2 \nonumber\\
&\stackrel{(c)}{\sim}& \frac{[g'(w)]^2 w(1-w)}{k^2 g(w)(1-g(w)) 2\ln 2} \nonumber\\
&\stackrel{(d)}{=}& \frac{1}{k^2 2\ln 2} J(w)
\end{eqnarray}
where (a) follows from (\ref{eqn:chi_approx}), (b) from (\ref{eqn:asym_gk}) and (\ref{eqn:qk_fk_dif}), (c) from Lemma \ref{lem:chernoff}, and (d) directly from the definitions in (\ref{eqn:G}) and (\ref{eqn:J}). The following theorem concludes what we have proved thus far:

\begin{theorem}\label{thm:I_asym}
Assume that Condition \ref{cond:cont} is satisfied, then
\begin{equation}
I^\joint_k(w,\bfp) \sim \frac{1}{k^2 2\ln 2} J(w), \quad \forall w \in (0,1).
\end{equation}
\end{theorem}

The joint fingerprinting game of (\ref{eqn:bin_minimax}) can now be approximated by the game with $J$ as its payoff function. We consider continuous probability distributions $f_W$ satisfying the following condition:

\begin{condition}\label{cond:f_int_finite}
The pdf $f_W$ is continuous on $[0,1]$ with
\begin{equation}
\int_0^1 \frac{dw}{f_W(w)w(1-w)} < \infty.
\end{equation}
\end{condition}

The following lemma shows the solution to the minimization problem with $J$ as its payoff fuction.
\begin{lemma}\label{lem:min_J}
Let $g(w)$ satisfy Condition \ref{cond:cont} and fix $f_W$ satisfying Condition \ref{cond:f_int_finite}. Then
\begin{equation}
\min_g \bbE_{f_W} \left[ J(W) \right] = \min_g \displaystyle\int_0^1 J(w) f_W(w)dw = \pi^2 \left[ \displaystyle\int_0^1 \frac{dw}{f_W(w)w(1-w)} \right]^{-1}
\end{equation}
where $J$ is defined in (\ref{eqn:J}). The minimum is achieved by
\begin{equation}\label{eqn:gopt}
g_\opt(w) = \frac{1}{2} \left[ 1- \cos \left( \frac{\pi\int_0^w \frac{dv}{f_W(v)v(1-v)}}{\int_0^1 \frac{dv}{f_W(v)v(1-v)}} \right) \right].
\end{equation}
\end{lemma}
\begin{proof}
We have
\begin{eqnarray}
\int J(w) f_W(w)dw &=& \int_0^1 w(1-w)[G'(w)]^2 f_W(w)dw \nonumber\\
&\stackrel{(a)}{\geq}& \dfrac{\left[\int_0^1 G'(w) dw \right]^2}{\int_0^1 \frac{dw}{w(1-w)f_W(w)}} \nonumber \\
&\stackrel{(b)}{=}& \pi^2 \left[ \displaystyle\int_0^1 \frac{dw}{f_W(w)w(1-w)} \right]^{-1}
\label{eqn:asym_J}
\end{eqnarray}
where (a) follows from the Cauchy-Schwarz inequality and (b) follows from the boundary conditions $G(0) = 0$ and $G(1) = \pi$ following directly from Condition \ref{cond:cont} and the definition of (\ref{eqn:G}). Equality holds in (a) when
\begin{equation}
G'_\opt(w) = \frac{\pi}{\int_0^1 \frac{dv}{f_W(v)v(1-v)}} \cdot \frac{1}{f_W(w)w(1-w)},
\end{equation}
which leads us to (\ref{eqn:gopt}) by (\ref{eqn:G}).
\end{proof}

\begin{corollary}
For $f_W$ satisfying Condition \ref{cond:f_int_finite}, we have
\begin{equation}\label{eqn:C_fW_asym}
C^\joint_k(f_W,\scrP_\marking) \sim \frac{\pi^2}{k^2 2\ln 2}  \left[ \displaystyle\int_0^1 \frac{dw}{f_W(w)w(1-w)} \right]^{-1}.
\end{equation}
\end{corollary}
\begin{proof}
This follows directly from Theorem \ref{thm:I_asym} and Lemma \ref{lem:min_J}.
\end{proof}

\begin{corollary}\label{cor:cap_lb}
\begin{equation}\label{eqn:cap_lb}
C^\joint_k(f_W^*,\scrP_\marking) \sim \frac{1}{k^2 2\ln 2}.
\end{equation}
\end{corollary}
\begin{proof}
The right-hand side of (\ref{eqn:C_fW_asym}) is maximized when $f_W = f_W^*$ by Lemma \ref{lem:lb_f_int}. Also note that by (\ref{eqn:gopt}) we have $g_\opt(w) = w$, which leads us to the interleaving attack.
\end{proof}

\begin{corollary}\label{cor:cap_ub}
\added{The interleaving attack is an ``equalizing strategy'' for the colluders that makes the payoff function $J(w)$ asymptotically independent of $w$:}
\begin{equation}\label{eqn:cap_ub}
I^\joint_k(w,\bfp^*) \sim \frac{1}{k^2 2\ln 2}, \quad \forall w \in (0,1).
\end{equation}
\end{corollary}
\begin{proof}
Let $g(w) = w$. Then $\bfp$ becomes the interleaving attack by (\ref{eqn:g_restrict}) and $J(w) \equiv 1$ by Definition \ref{def:G_J}. \removed{This shows that the interleaving attack is the ``equalizing strategy'' for the colluders which makes the payoff function asymptotically flat.}
\end{proof}

\begin{corollary}
The fingerprinting capacity under the marking assumption satisfies
\begin{equation}
C^\joint_k(\scrP_W,\scrP_\marking) \sim \dfrac{1}{k^2 2 \ln2}.
\end{equation}
Furthermore, the arcsine distribution $f_W^*$ and the interleaving attack $\bfp^*$ are the respective maximizing and minimizing strategies that achieve the asymptotic capacity value.
\end{corollary}
\begin{proof}
The asymptotic relations (\ref{eqn:cap_lb}) and (\ref{eqn:cap_ub}) establish matching asymptotic lower and upper bounds on $C^\joint_k(\scrP_W,\scrP_\marking)$ respectively.
\end{proof}

\subsection{Numerical Results for Beta Distributions}\label{ssec:Beta}

\begin{figure*}[t]

\begin{minipage}{.32\textwidth}
  \centering
  \includegraphics[width=.9\textwidth]{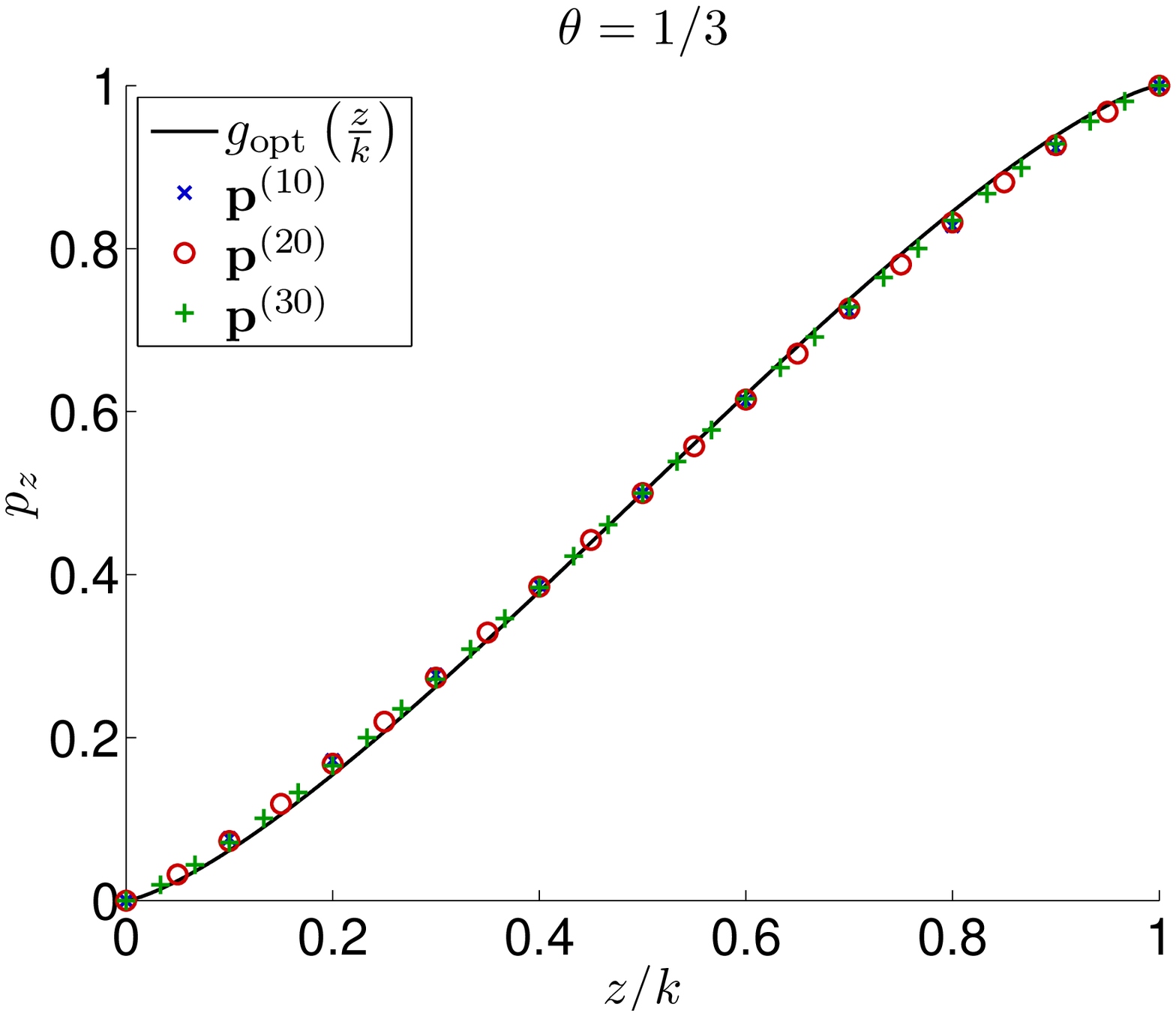}
\end{minipage}
\hfill
\begin{minipage}{.32\textwidth}
  \centering
  \includegraphics[width=.9\textwidth]{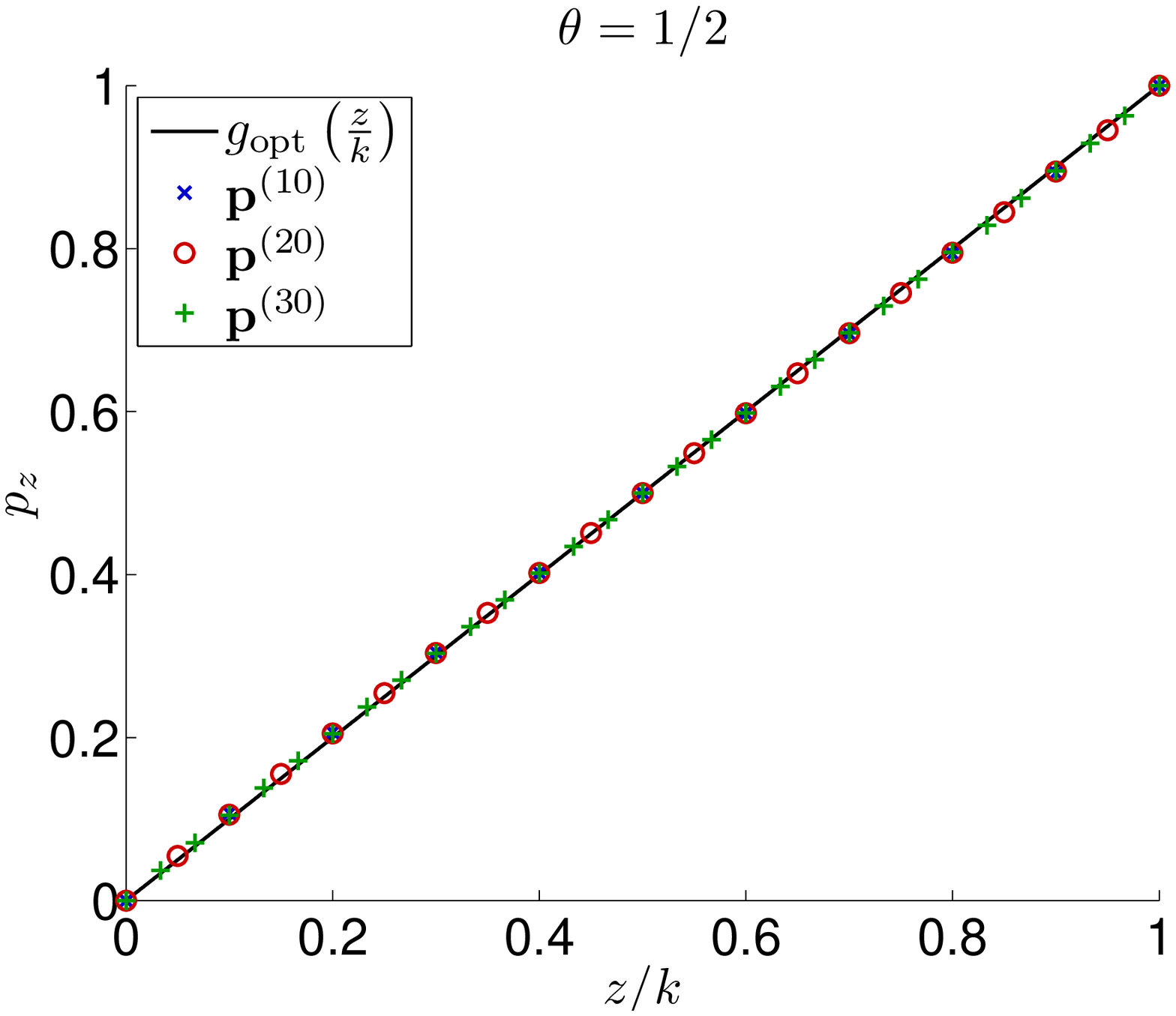}
\end{minipage}
\hfill
\begin{minipage}{.32\textwidth}
  \centering
  \includegraphics[width=.9\textwidth]{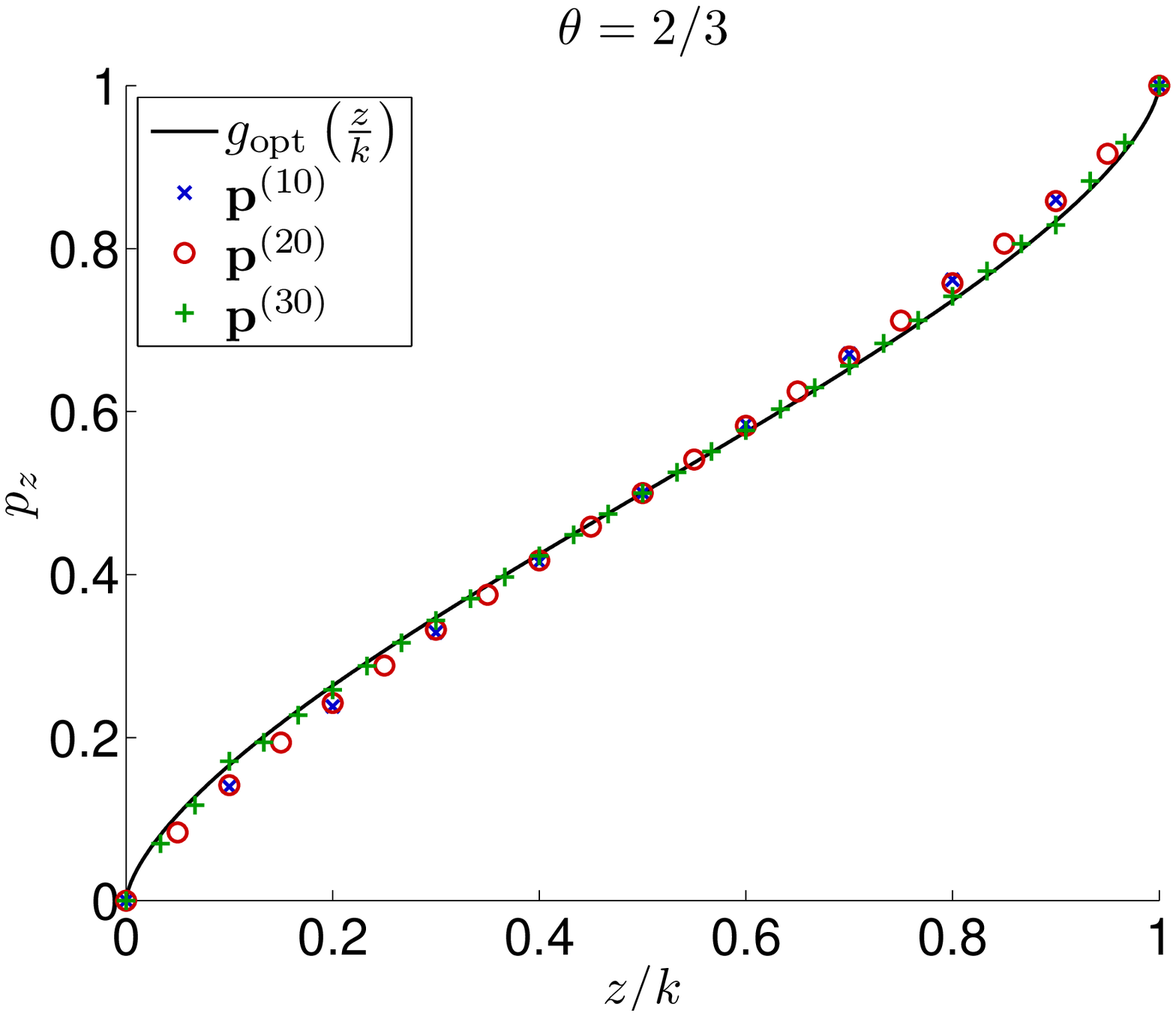}
\end{minipage}

\caption{$g_{\opt}(z/k)$ and minimizing collusion channels ${\bfp^{(k)}}$ for $k = 10$, $20$, and $30$ and $P_W = \Beta(\theta,\theta)$}
\label{fig:p_plot_Beta_mark}

\end{figure*}

\begin{figure*}[t]

\begin{minipage}{.32\textwidth}
  \centering
  \includegraphics[width=.9\textwidth]{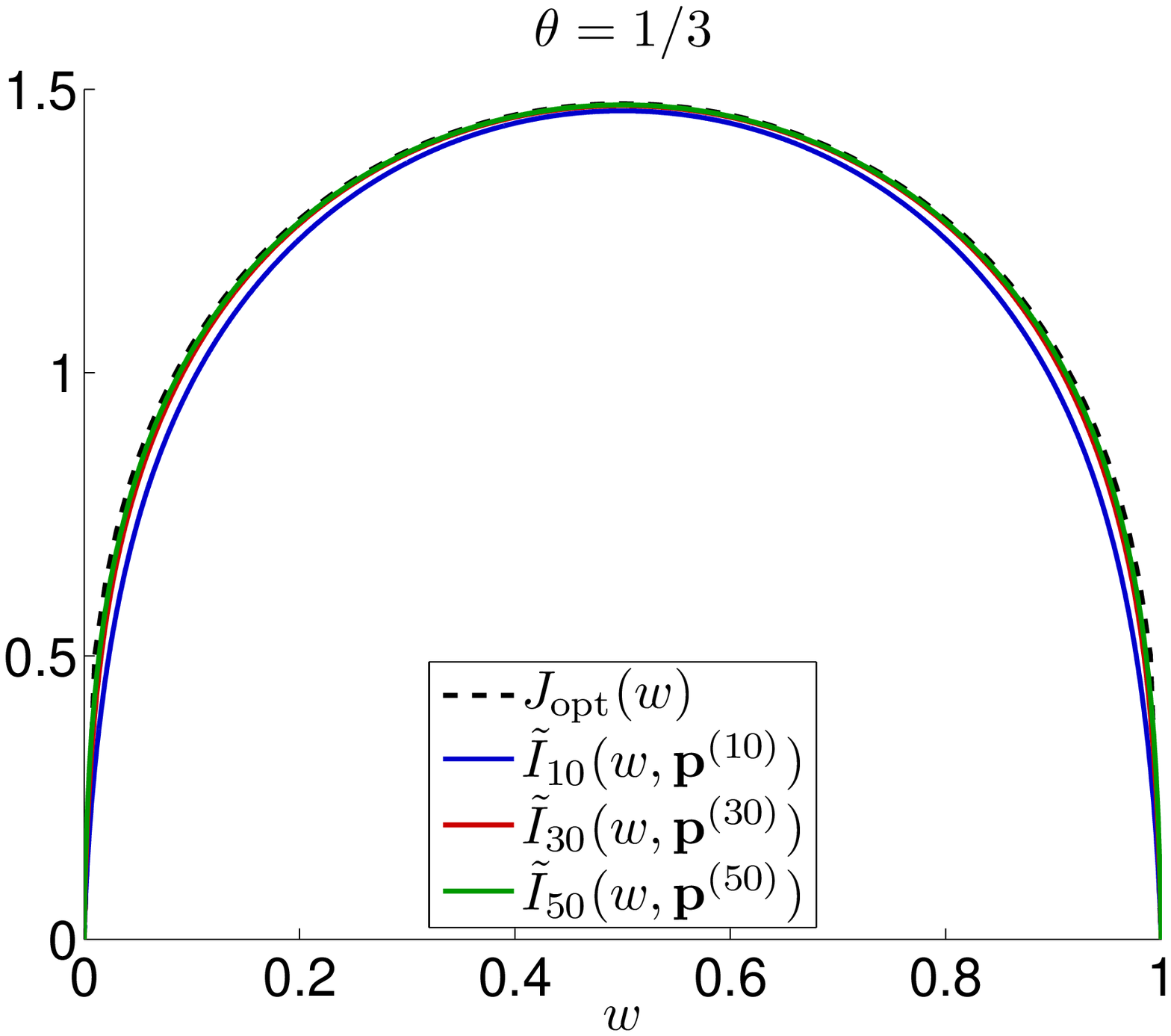}
\end{minipage}
\hfill
\begin{minipage}{.32\textwidth}
  \centering
  \includegraphics[width=.9\textwidth]{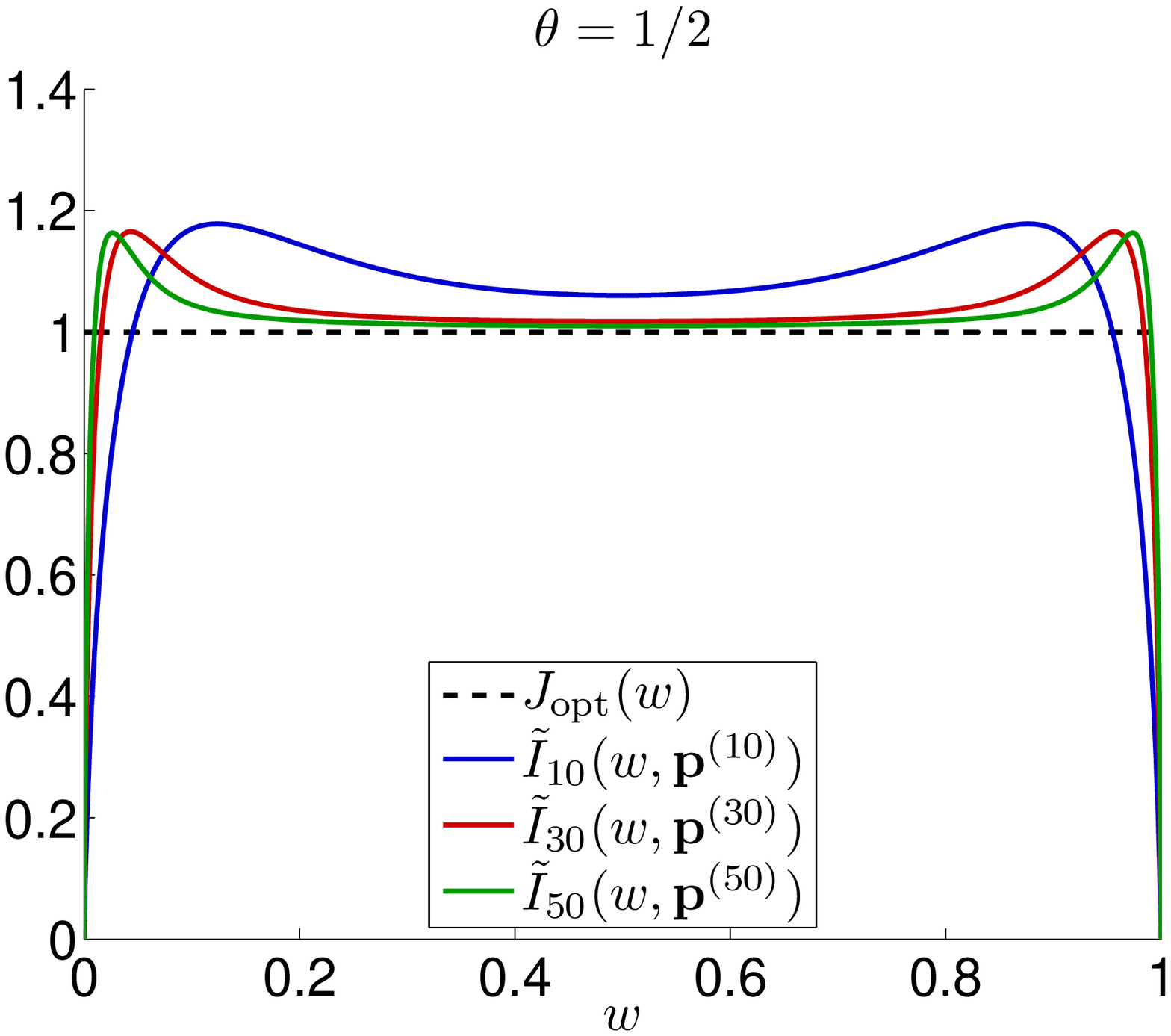}
\end{minipage}
\hfill
\begin{minipage}{.32\textwidth}
  \centering
  \includegraphics[width=.9\textwidth]{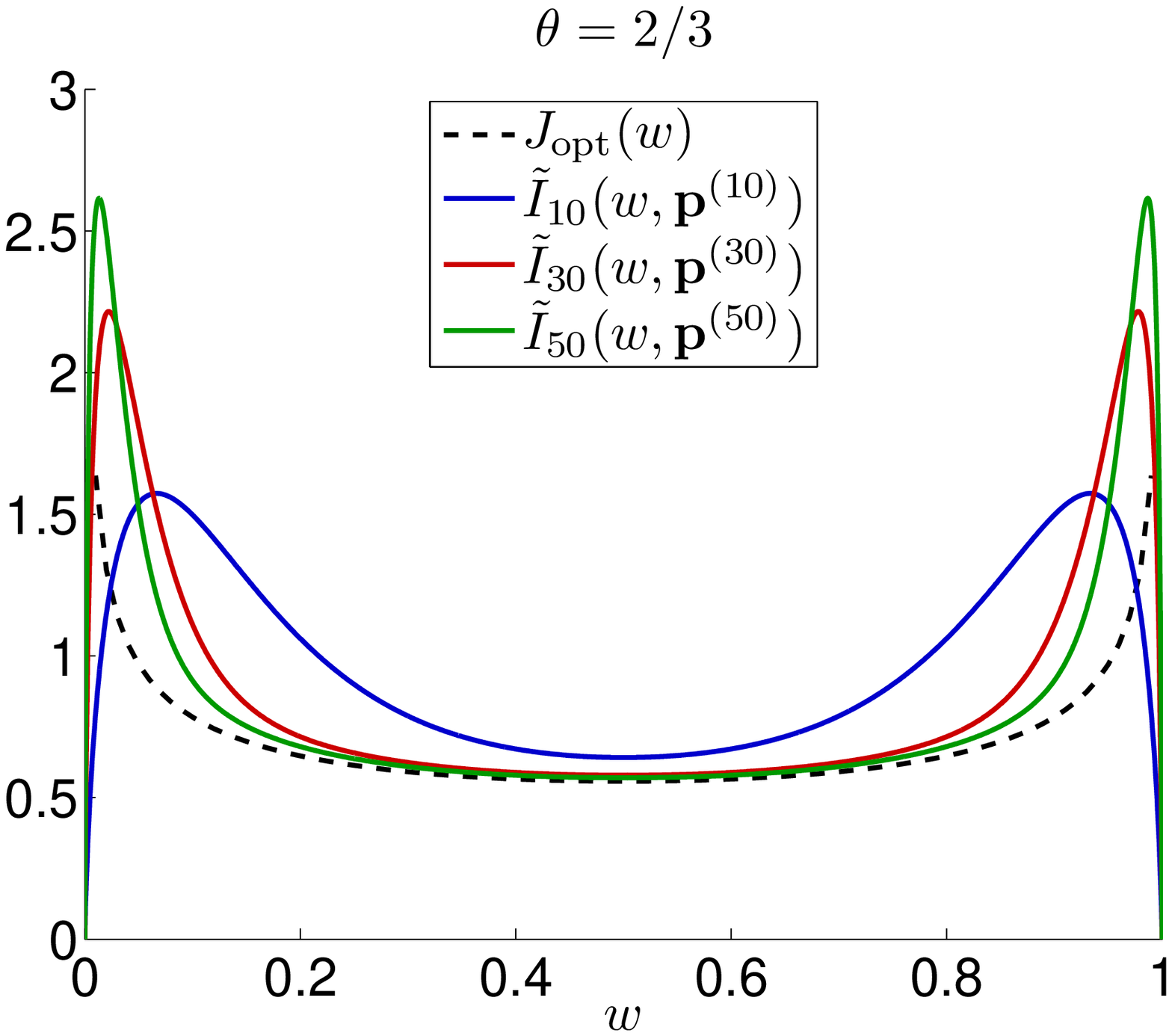}
\end{minipage}

\caption{$J_\opt(w)$ and normalized payoff function $\tilde{I}_k(w,\bfp^{(k)})$ for $k = 10$, $30$, and $50$ and $P_W = \Beta(\theta,\theta)$}
\label{fig:I_w_fixedp}

\end{figure*}

We now use the family of Beta distributions to illustrate the asymptotic\changed{s}\removed{ behavior} of the joint fingerprinting game. Let $f^\theta_W$ be the pdf defined in (\ref{betapdf}). Condition \ref{cond:f_int_finite} is satisfied for any $\theta \in (0,1)$. For $\theta = 1/3, 1/2$, and $2/3$, we find the minimizing collusion channels $\bfp^{(k)}$ for $f^\theta_W$ and compare them with $g_\opt(z/k)$ obtained by (\ref{eqn:gopt}). Fig. \ref{fig:p_plot_Beta_mark} shows that $p_z^{(k)}$ does converge to $g_\opt(z/k)$ as $k \to \infty$ as expected, which also rationalizes our assumption of Condition \ref{cond:cont}.

Consider the normalized payoff function $\tilde{I}^\joint_k(w,\bfp) \triangleq k^2 2\ln 2 \cdot I^\joint_k(w,\bfp)$, which by Theorem \ref{thm:I_asym} is asymptotically close to $J(w)$. Suppose $J_\opt$ is obtained by substituting $g_\opt$ of (\ref{eqn:gopt}) into (\ref{eqn:G}) and (\ref{eqn:J}). Again for $\theta = 1/3, 1/2$, and $2/3$, we compare $\tilde{I}^\joint_k(w,\bfp^{(k)})$ with $J_\opt(w)$ in Fig. \ref{fig:I_w_fixedp}. As shown in the figure, $\tilde{I}^\joint_k(w,\bfp^{(k)})$ is asymptotically flat over $(0,1)$ when $\theta = 1/2$, which is the case when $f^\theta_W$ is chosen properly. If $\theta < 1/2$, which means $f^\theta_W$ has too much weight around 0 and 1, then the colluders' choice of $g_\opt$ makes $J_\opt$ peak at $w = 1/2$. If on the contrary $\theta > 1/2$, then too much weight around $1/2$ is put on $f^\theta_W$, and $g_\opt$ makes $J_\opt$ peaks at $w = 0,1$.

\subsection{Why Are the Arcsine Distribution and the Interleaving Attack Optimal for Large Coalitions?}\label{ssec:bin_Bayes}

\added{The analysis and the numerical results above show not only that the asymptotic capacity is $(k^2 2 \ln 2)^{-1}$, but also that both the arcsine distribution for the maximizer and the interleaving attack for the minimizer achieve the same asymptotic value. Such results are very interesting, and at the same time raise some issues for further investigation. One topic concerns the regularity constraint (Condition \ref{cond:cont})upon which the asymptotic analysis in Sec. \ref{ssec:asym_anal} is based. However, it is reasonable to conjecture that the same asymptotics hold without the regularity condition. Our numerical results indeed suggest this condition may not be necessary. Moreover, it is important to mention that both the asymptotic lower bound on capacity (see \cite[Theorem 15]{Amiri2010}) and Corollary \ref{cor:cap_ub} (which contributes to the asymptotic upper bound) hold without the regularity constraint.}

\added{Both the arcsine distribution and the interleaving attack have been \added{extensively} studied in the literature.} In 2003, \changed{Tardos applied the arcsine distribution to fingerprinting}\removed{the arcsine distribution was first applied to the applications of fingerprinting by Tardos} \cite{Tardos2003}. How he fine-tuned his codes, however, had been a mystery until \v{S}kori\'{c} et al. \cite{vSkori'c2008a} and Furon et al. \cite{Furon2008} rationalized Tardos' choices based on Gaussian approximations. On the other hand, the interleaving attack has been frequently adopted to model the collusion channel in the literature \cite{Anthapadmanabhan2008,Furon2008,Furon2009a}, but no conceptual reasoning has been proposed on why it should be the coalition's optimal choice. Fortunately, owing to the discovery of the capacity formulas (Theorem \ref{thm:single_game}), we can now study fingerprinting games from the information-theoretic point of view. In the previous subsection, \added{we established} the optimality of these two strategies \removed{is established} based on asymptotic methods. \changed{Here we provide}\removed{In this subsection} a statistical interpretation\removed{ is provided}.

We may think of the (joint) fingerprinting \added{capacity} game as follows: the coalition is given $k$ independent \emph{observations} $X_1, \ldots, X_k$ distributed according to an unknown distribution $\Bernoulli(W)$ chosen at random by the fingerprinting embedder from the \changed{family}\removed{set} $\{\Bernoulli(W): W \in [0,1] \}$ according to a known \emph{prior distribution} $P_W$. Upon generating $Y$ according to \changed{the conditional distribution $p_{Y|Z}$ given}\removed{$\Bernoulli(p_Z)$ based on} the sufficient statistic $Z = \sum_{i=1}^k X_i$, the coalition suffers a loss $I(Z;Y|W = w)$. The \emph{risk} of the game \added{$I(Z;Y|W)$} is the average loss under $P_W$.

As emphasized in previous works \cite{Moulin2008b,Furon2009a}, the choice of the embedding distribution $P_W$, or \emph{prior selection} in statistician's language, is crucial to the fingerprinting game. If no randomization takes place \cite{Somekh-Baruch2007}, or equivalently, if the realization $w$ is revealed to the pirates \cite{Furon2009a}, then the maximin game value decays exponentially with coalition size $k$ (see (\ref{eqn:C_expo})). Loosely speaking, the loss the pirates suffer is mainly due to their error in estimating $W$. If they have a good estimation of the time-sharing random variable $W$, then the loss they suffer is small.

Jeffreys' prior \cite{Jos`eM.Bernardo2000} is a ``non-informative'' prior that \changed{plays an important role in Bayesian statistics}\removed{is central to the study of Bayesian analysis}. Given a family of distributions with an unknown parameter, Jeffreys' prior is proportional to the square root of the Fisher information. Conceptually speaking, Jeffreys' prior is the ``least-favorable'' distribution in regard to estimating that parameter. For the Bernoulli trial with the probability of success $w$ as parameter, the Fisher information is $I(w) = \left[w(1-w)\right]^{-1}$ and thus Jeffreys' prior is 
\begin{eqnarray}
f(w) \propto \frac{1}{\sqrt{w(1-w)}}
\end{eqnarray}
which is exactly the arcsine distribution! \removed{The pdf is U-shaped, which suggests that $W$ is easier to estimate around $1/2$ than close to 0 and 1. The selection of Jeffreys' prior equalizes such difficulty discrepancy.}

The optimality of the interleaving attack is closely related to the concept of ``equalizer rule'' in game theory. From Corollary \ref{cor:cap_ub}, \changed{interleaving}\removed{it} is the asymptotic equalizing strategy, which is the desirable attribute we expect for a saddle-point solution. The optimal collusion channel depends on the prior by (\ref{eqn:gopt}), and from the proof of Corollary \ref{cor:cap_lb}\added{,} the interleaving attack is optimal under \changed{the arcsine distribution}\removed{Jeffreys' prior}. \changed{Also observe}\removed{Another interpretation is} that the interleaving attack is the strategy where the colluders generate $Y$ according to $\Bernoulli(\widehat{W})$, where $\widehat{W} = Z/k$ is the maximum likelihood estimator of $W$, which is asymptotically unbiased (as $k \rightarrow \infty$) and has minimum asymptotic variance (equal to $(kI(w))^{-1} = w(1-w)/k$).

\added{Another interesting question is what are the asymptotics of the simple fingerprinting game. In Corollary \ref{cor:simplebd}, we established upper and lower bounds on $C^{\simple}_k$. Note that the upper bound is obtained by assuming the interleaving attack for the coalition and it coincides with the asymptotic rate $(k^2 2\ln 2)^{-1}$ of $C^\joint_k$. On the other hand, Fig. \ref{fig:p_plot_TS_mark}(b) indicates that the optimal attack is actually quite different from the interleaving attack. This suggests that the pirates can exploit the suboptimality of the single-user decoder and perform a stronger attack. The study of the exact asymptotics of the simple fingerprinting game, is left as future work.}\looseness=-1   
\section{Summary}\label{sec:sum}

In this work, we proved new upper and lower bounds on the maximum achievable rate of binary fingerprinting codes for arbitrary coalition size by studying the minimax and the maximin fingerprinting games. We also provided asymptotic approximations of the capacity as well as both the fingerprinting embedder and the coalition's strategies. The results suggest that fingerprinting games under the Boneh-Shaw marking assumption have a close relation to the Fisher information and Jeffreys' prior for the Bernoulli model.
\appendix[Derivation of $C^\simple_k(\scrP_W,\bfp^*)$]

The function $I^\simple_k(w,\bfp^*)$ is indeed symmetric around $w = 1/2$ and has a global maximum at $w = 1/2$ as suggested by numerical experiments in \cite{Furon2009a}. We first prove the following lemma:

\begin{lemma}\label{lem:d_ineq}
For $r \geq s \geq 0$ and $r + s \leq 1$, we have
\begin{equation}
d(r \parallel s) \geq d(s \parallel r).
\end{equation}
\end{lemma}
\begin{proof}
The difference $\delta(r,s)$ between the two sides is
\begin{eqnarray}
\delta(r,s) &=& d(r \parallel s) - d(s \parallel r)\nonumber\\
&=& r \log \frac{r}{s} + (1-r) \log \frac{1-r}{1-s} - s \log \frac{s}{r} - (1-s) \log \frac{1-s}{1-r}\nonumber\\
&=& (r+s) \log \frac{r}{s} + (2-r-s) \log \frac{1-r}{1-s}.
\end{eqnarray}
Then
\begin{equation}
\frac{\partial}{\partial s} \delta(r,s) = \frac{s-r}{s(1-s)\ln 2} + \log \frac{r(1-s)}{s(1-r)}
\end{equation}
and
\begin{equation}
\frac{\partial^2}{\partial s^2} \delta(r,s) = \frac{(r-s)(1-2s)}{s^2(1-s)^2 \ln 2} \geq 0
\end{equation}
for all $r \geq s$ and $r+s \leq 1$. Now when $s \leq r \leq 1/2$, we have $\left.\frac{\partial}{\partial s} \delta(r,s)\right|_{s=r} = \delta(r,r) = 0$. Thus $\frac{\partial}{\partial s} \delta(r,s) \leq 0$ and thus $\delta(r,s) \geq 0$. When $1/2 \leq r \leq 1-s$, we have $\left.\frac{\partial}{\partial s} \delta(r,s)\right|_{s=1-r} \leq 0$ and $\delta(r,1-r) = 0$. Hence similarly $\frac{\partial}{\partial s} \delta (r,s) \leq 0$ and hence $\delta(r,s) \geq 0$. To prove the inequality $\left.\frac{\partial}{\partial s} \delta(r,s)\right|_{s=1-r} \leq 0$, let 
$$\lambda(r) \triangleq \left.\frac{\partial}{\partial s} \delta(r,s)\right|_{s=1-r} = \frac{1-2r}{r(1-r)\ln 2} + 2 \log \frac{r}{1-r}$$
and since $\lambda(1/2) = 0$ and $\lambda'(r) = -\frac{(1-2r)^2}{r^2(1-r)^2 \ln 2} \leq 0$ it follows that $\lambda(r) \leq 0$ for all $r \in [1/2,1]$.
\end{proof}

Now the payoff function can be written as
\begin{eqnarray}
I^\simple_k(w,\bfp^*) &\stackrel{(a)}{=}& wd({\bfal^1}'\bfp^*\parallel\bfal'\bfp^*)+(1-w)d({\bfal^0}'\bfp^*\parallel\bfal'\bfp^*)\nonumber\\
&\stackrel{(b)}{=}& wd(w+\frac{1-w}{k}\parallel w)+(1-w)d(w-\frac{w}{k}\parallel w)\label{eqn:I_simple_inter}
\end{eqnarray}
where (a) follows from (\ref{eqn:I_simple}) and (b) follows from (\ref{eqn:interleaving}) and Lemma \ref{lem:der_iden}, and by which we can easily verify the symmetry property $I^\simple_k(w,\bfp^*) = I^\simple_k(1-w,\bfp^*)$. Hence it suffices to show that $I^\simple_k(w,\bfp^*)$ is nondecreasing for $0 \leq w \leq 1/2$. Taking the derivative of (\ref{eqn:I_simple_inter}) with respect to $w$ and after some simplifications, we obtain
\begin{equation}
\frac{\partial}{\partial w} I^\simple_k(w,\bfp^*) = d(w^+ \parallel w^-) - d(w^- \parallel w^+)
\end{equation}
where $w^+ \triangleq w+\frac{1-w}{k}$ and $w^- \triangleq w-\frac{w}{k}$. By Lemma \ref{lem:d_ineq} it follows that $\frac{\partial}{\partial w} I^\simple_k(w,\bfp^*) \geq 0$ for all $w \in [0,1/2]$ and hence $I^\simple_k(w,\bfp^*)$ achieves its maximum at $w = 1/2$.   

\section*{Acknowledgment}

The authors would like to thank Ehsan Amiri for illuminating discussions and helpful comments.

\ifCLASSOPTIONcaptionsoff
  \newpage
\fi



%

\bibliographystyle{IEEEtran}
\bibliography{IEEEabrv,ref}

%







\end{document}